\documentclass[12pt]{article}%
\usepackage{amsfonts}
\usepackage[onehalfspacing]{setspace}
\usepackage{amsmath}
\usepackage{amssymb}
\usepackage{graphicx}
\usepackage{color}
\usepackage{float}
\usepackage[authoryear]{natbib}%
\providecommand{\U}[1]{\protect\rule{.1in}{.1in}}
\newtheorem{theorem}{Theorem}

\newtheorem{corollary}{Corollary}

\newtheorem{definition}[theorem]{Definition}

\newtheorem{proposition}{Proposition}
\newtheorem{remark}{Remark}

\newtheorem{step}{Step}[section]

\newenvironment{proof}[1][Proof]{\noindent\textbf{#1.} }{\ \rule{0.5em}{0.5em}}

\begin{document}

\title{False Narratives and Political Mobilization\thanks{Spiegler acknowledges
financial support from ERC Advanced Investigator grant no. 692995. We thank
Tuval Danenberg, Danil Dmitriev, Federica Izzo, Gilat Levy and Nathan Hancart
for helpful comments.}}
\author{Kfir Eliaz, Simone Galperti, and Ran Spiegler\thanks{Eliaz: Tel Aviv
University and the University of Utah. Galperti: UC San Diego. Spiegler: Tel
Aviv University and University College London.}}
\maketitle

\begin{abstract}
We present an equilibrium model of politics in which political platforms
compete over public opinion. A platform consists of a policy, a coalition of
social groups with diverse intrinsic attitudes to policies, and a narrative.
We conceptualize narratives as subjective models that attribute a commonly
valued outcome to (potentially spurious) postulated causes. When quantified
against empirical observations, these models generate a shared belief among
coalition members over the outcome as a function of its postulated causes. The
intensity of this belief and the members' intrinsic attitudes to the policy
determine the strength of the coalition's mobilization. Only platforms that
generate maximal mobilization prevail in equilibrium. Our equilibrium
characterization demonstrates how false narratives can be detrimental for the
common good, and how political fragmentation leads to their proliferation. The
false narratives that emerge in equilibrium attribute good outcomes to the
exclusion of social groups from ruling coalitions.

\end{abstract}

\section{Introduction}

Success in democratic politics requires the mobilization of public opinion,
which takes various forms: rallies, petitions, social media activism, and
ultimately voter turnout. Fluctuations in public opinion can determine which
policies get implemented and which coalitions of social groups form around
them (\cite{Burstein2003}). In turn, opinion makers (politicians, news
outlets, pundits) use past performance of policies and coalitions as raw
material for shaping public opinions. This paper is an attempt to shed light
on this interplay.

Our starting point is the idea that \textit{narratives} are a powerful tool
for mobilizing public opinion. This is a familiar idea with numerous
expressions in academic and popular discourse. After Senator John Kerry lost
the 2004 presidential elections, his political strategist Stanley Greenberg
said that \textquotedblleft a narrative is the key to
everything\textquotedblright\ and that Republicans had \textquotedblleft a
narrative that motivated their voters.\textquotedblright\footnote{See William
Safire's New York Times article titled \textquotedblleft
Narrative\textquotedblright%
\ (https://www.nytimes.com/2004/12/05/magazine/narrative.html).}
\cite{SJM2011} write: \textquotedblleft Policy narratives are the lifeblood of
politics. These strategically constructed `stories' contain predictable
elements and strategies whose aim is to influence public opinion toward
support for a particular policy preference.\textquotedblright

We build on the approach of \cite{ES2020}, who proposed to view political
narratives as \textit{causal models }that attribute public outcomes (such as
national security, economic growth, or the climate) to other factors (such as
policies or external variables). A narrative generates a probabilistic belief
by being \textquotedblleft estimated\textquotedblright\ against the observed
correlation between the outcome and its postulated causes. The stronger this
correlation, the stronger the belief that the narrative generates, and hence
the extent to which the narrative mobilizes agents behind a political
platform. This means that competition among social groups for political power
is to some extent a battle between conflicting narratives over what drives
public outcomes. Under this modeling approach, a \textit{false narrative}
corresponds to a \textit{misspecified} causal model. It can produce wrong
beliefs because it imposes an incorrect causal meaning on the correlation it highlights.

While \cite{ES2020} assumed a representative agent, this paper considers a
heterogenous society consisting of multiple social groups. We think of a
social group as a collection of agents having shared ideological,
socio-economic or ethnic/religious characteristics, as well as a distinct
political representation. For example, society can be divided into a left and
a right---each having its own representative political party. Another example
is the distinction between Flemish and French parties in Belgium, or the
various religious parties in Israel. The departure from a representative-agent
environment means that narratives can now also attribute outcomes to which
social groups are in power. This generates new insights into the role of
various kinds of false narratives in determining the structure of ruling
coalitions---namely, which groups are represented in government---and
ultimately the implemented policies.

Our model makes the stark, simplifying assumption that policies are the
\textit{only true cause} of public outcomes. However, ideological differences
between social groups will naturally give rise to correlations between the
structure of ruling coalitions, the policies they implement, and these
policies' resulting outcomes. A false narrative can exploit this correlation
and causally attribute the outcome \textit{solely} to social groups' power
status (i.e., whether they belong to the ruling coalition), even though in
reality this correlation is due to confounding by the implemented policies.

For illustration, suppose that a certain coalition $C$ systematically refrains
from taxing wealth. As a result, social inequality will tend to rise when $C$
is in power. A rival coalition $C^{\prime}$ may exploit this correlation and
spin a false narrative that in order to reduce inequality, we only need to
keep the social groups in $C$ out of power. Because this narrative does not
attribute the outcome to its true cause (namely, tax policy), it enables
$C^{\prime}$ to gain support: On the one hand, $C^{\prime}$ acts exactly like
$C$ by not proposing an unpopular wealth tax; on the other hand, it can claim
that by elbowing out $C$ it is doing something to lower inequality, which $is$
popular. Thus, in a sense, $C^{\prime}$ uses $C$ as a \textquotedblleft
scapegoat.\textquotedblright\ Our main objective in this paper is to
understand how such false narratives can gain ascedance, what form they may
take, and how they shape public policies and ruling coalitions.

In our setting, a \textit{political platform} consists of a policy, a
coalition of social groups, and a narrative. Given a long-run joint empirical
distribution over prevailing platforms and public outcomes, different
narratives may induce conflicting beliefs regarding the consequences of
policies and coalitions. Moreover, changes in long-run frequencies may change
the beliefs induced by a narrative, and therefore the extent to which it can
mobilize a social group. We define an \textit{equilibrium} as a probability
distribution over platforms, such that every platform in its support maximizes
the total mobilization of the social groups belonging to the platform's coalition.

This definition captures in reduced form the idea that a platform's success
depends on the strength of its popular support (in terms of the number and
size of social groups that rally behind it, and the intensity with which they
do so). It does so in the spirit of competitive equilibrium, as in
\cite{rothschild1976equilibrium}. The backstory is that there is
\textquotedblleft free entry\textquotedblright\ of office-motivated political
entrepreneurs who propose policy-narrative combinations. If a particular
combination attracts stronger support than the current combination, it will
overthrow it, such that the platform that eventually prevails is the one that
maximizes total support. We present a convergence result that provides a
dynamic foundation for the equilibrium concept.

In line with the competitive-equilibrium-like approach, we do not offer an
explicit model of the political process. In particular, there will be no
explicit game between platforms. This is in keeping with our interest in
competition over public opinion in a broad sense, rather than in specific
electoral procedures or post-election coalitional bargaining. However, in the
case of two social groups, our equilibrium concept is consistent with a
two-party voting model with endogenous voter turnout.

Using this formalism, we obtain several qualitative insights. First, in
addition to the true narrative that attributes outcomes to policies, two types
of false narratives emerge in equilibrium. The first type is a
\textquotedblleft denialist\textquotedblright\ narrative that attributes
outcomes to external forces outside society's control. The other type is a
\textquotedblleft tribal\textquotedblright\ narrative that attributes a good
public outcome to the \textit{exclusion} of some social groups from the ruling
coalition. In a political speech or a social-media post, such a narrative
could appear as \textquotedblleft national security was in good shape when the
Left was out of power.\textquotedblright\ 

Recent public debates over rising inflation are suggestive of these kinds of
narratives: Much of the action in these debates involves competing claims over
the causes of high inflation. Some narratives attributes it to government
actions (fiscal expansion), others to external factors (global supply-chain
disruptions), and yet others assign credit or blame for the level of inflation
solely to the party in power, without attempting to link this to the party's
policy decisions. A selection of press quotes demonstrates the form that these
conflicting narratives take:\medskip

\begin{quote}
\textquotedblleft As prices have increased\,...\,some Democrats have landed on
a new culprit: price gouging\,...\,For Democrats, it is a convenient
explanation as inflation turns voters against President Biden. It lets
Democrats deflect blame from their pandemic relief bill, the American Rescue
Plan, which experts say helped increase prices.\textquotedblright%
\footnote{https://www.nytimes.com/2022/06/14/briefing/inflation-supply-chain-greedflation.html}%
\medskip

\textquotedblleft Democrats have blamed supply chain deficiencies due to
COVID-19, as well as large corporations and monopolies.\textquotedblright%
\footnote{https://fivethirtyeight.com/features/what-democrats-and-republicans-get-wrong-about-inflation/}%
\medskip

\textquotedblleft As the midterm elections draw nearer, a central conservative
narrative is coming into sharp focus: President Biden and the
Democratic-controlled Congress have made a mess of the American
economy.\textquotedblright%
\footnote{https://www.nytimes.com/2022/06/11/opinion/fed-federal-reserve-inflation-democrats.html}%
\medskip
\end{quote}

The distinction between a false narrative that attributes outcomes to whoever
is in power and a more accurate narrative that attributes outcomes to policies
is also alluded to by Paul Krugman in a recent article about the politics of
inflation:\medskip

\begin{quote}
\textquotedblleft...\,voters aren't saying, \textquotedblleft Trimmed mean
P.C.E. inflation is too high because fiscal policy was too
expansionary.\textquotedblright\ They're saying, \textquotedblleft Gas and
food were cheap, and now they're expensive\,...\textquotedblright\ So when
people say --- as they do --- that gas and food were cheaper when Donald Trump
was president, what do they imagine he could or would be doing to keep them
low if he were still in office?\textquotedblright%
\footnote{https://www.nytimes.com/2022/06/02/opinion/inflation-biden.html. See
also \cite{Weaver2013} and \cite{SHZ2017}.}\medskip
\end{quote}

We wish to emphasize that we do not argue that our specific model matches the
inflation scenario. Nevertheless, we believe it offers an insight into the
interplay between the popularity of various types of false narratives, and the
objective statistical reality that both feeds the narratives and gets shaped
by them through the policy choices that different narratives promote.

Our second qualitative insight is that the false narratives that are employed
in equilibrium sustain a policy that would not be taken if the only prevailing
narrative were a true one (which correctly attributes outcomes to policies).
The function of false narratives is to resolve the cognitive dissonance
between the objective ineffectiveness of a policy for the common good and that
same policy's intrinsic appeal for a social group (e.g., based on costs or
benefits which affect that group). This is achieved by deflecting
responsibility for the public outcome from its true cause to spurious causes.
Moreover, when society's structure of political representation is more
fragmented and \textquotedblleft tribal\textquotedblright\ (in a sense we make
precise in our model), false narratives proliferate and lead to further
crowding out of the true narrative and the policy it justifies.

Finally, we characterize the structure of coalitions that form in equilibrium.
False narratives can give rise to coalitions that would not form if only the
true narrative prevailed. In particular, when a political platform employs a
\textquotedblleft tribal\textquotedblright\ narrative, it may
\textquotedblleft scapegoat\textquotedblright\ social groups that support that
platform's policy (indeed, they implement the same policy when they are in
power), yet their exclusion from the platform's coalition is necessary for the
narrative's effectiveness. We also perform comparative statics with respect to
the polarization of attitudes toward policies in society. We show that greater
polarization can be detrimental for the common good through the proliferation
of false narratives.

\section{A Model}

We begin by introducing the building blocks of our model.\medskip

\noindent\textit{Policies and outcomes}

\noindent Our model examines public-opinion battles over a single salient
issue, defined by a \textit{public outcome} variable $y\in\{0,1\}$. Assume
there is a social consensus that $y=1$ is a desirable outcome, and refer to it
accordingly as the \textquotedblleft good\textquotedblright\ outcome. For
example, the issue can be economic growth such that $y=1$ represents high
growth. Let $a\in A=\{\ell,h\}$ be a \textit{policy}. Policies cause outcomes
according to the following objective conditional probability distribution:%
\begin{equation}
\Pr(y=1\mid a)=\left\{
\begin{array}
[c]{ccc}%
q & \text{if} & a=h\\
0 & \text{if} & a=\ell
\end{array}
\right.  \label{Pr(y|a)}%
\end{equation}
where $q\in(0,1]$. \medskip

\noindent\textit{Social groups and coalitions}

\noindent Let $N=\{1,...,n\}$ be a set of \textit{social groups}. A
\textit{coalition} is a non-empty, strict subset of groups $C\subset N$. Note
that this rules out only the possibility of a ``grand coalition'' $C=N$
including \textit{all} groups (which can be achieved via other more primitive
assumptions as we do in Section~\ref{sec: fragmented_societies}).\medskip

The following notation will be useful. Let $x=(x_{0},...,x_{n})$ be a profile
of binary variables, where $x_{0}\in\{\ell,h\}$ and $x_{i}\in\{0,1\}$ for
every $i>0$. For every $S\subset\{0,...,n\}$, denote $x_{S}=(x_{i})_{i\in S}$.
Define the following function that assigns values of $x$ to policy-coalition
pairs: for every $(a,c)$, $x_{0}(a,C)=a$; and for every $i>0$, $x_{i}(a,C)=1$
if and only if $i\in C$. For $i>0$, we will refer to $x_{i}$ as group $i$'s
\textquotedblleft\textit{power status}.\textquotedblright\medskip

\noindent\textit{Narratives}

\noindent We define a narrative as a subset $S\subseteq\{0,...,n\}$, which is
the set of indices of the components of $x$. The set $S$ defines the
components to which the narrative attributes the public outcome $y$. For
example, $S=\{0,2\}$ means that the narrative attributes the outcome to the
policy and the power status of social group $2$. Given a long-run joint
probability distribution $p$ over $(x,y)$---which we endogenize below---a
narrative $S$ generates a probabilistic belief over the outcome conditional on
its postulated causes, given by $(p(y\mid x_{S}))$.\footnote{We use the
abbreviated notation $(p(y\mid x_{S}))=(p(y\mid x_{S}))_{x_{S},y}$.} In other
words, the narrative draws attention to the correlation between $y$ and
$x_{S}$ and imposes a causal meaning on this correlation.

Two narratives have a special status in our model. First, $S=\{0\}$ is the
\textquotedblleft\textit{true narrative}\textquotedblright\ because it
attributes $y$ to its sole true cause $a$. Every other narrative is false in
the sense that it attributes $y$ to wrong causes. Second, $S=\varnothing$ is a
\textquotedblleft\textit{denialist}\textquotedblright\ narrative that does not
attribute $y$ to any of the endogenous variables. In particular, it claims
that the policy has no consequences. We interpret the denialist narrative as
the attribution of outcomes to external variables.

A third type of narrative will also play a key role in our model. These are
narratives $S\subseteq N$, where $S\neq\varnothing$. We refer to such
narratives as \textquotedblleft\textit{tribal}\textquotedblright\ because they
attribute outcomes to the power status of social groups, without mentioning
policies.\medskip

\noindent\textit{Platforms}

\noindent We refer to policy-coalition-narrative triples $(a,C,S)$ as
\textit{platforms}. These are the objects that vie for popularity in our
model. Given an objective long-run probability distribution $\sigma$ over
platforms $(a,C,S)$, we can define an induced joint distribution over
$(a,C,S,y)$ that incorporates expression (\ref{Pr(y|a)}):%
\[
p_{\sigma}(a,C,S,y)=\sigma(a,C,S)\cdot\Pr(y\mid a).
\]
Hereafter, we denote the support of $\sigma$ by $Supp(\sigma)$. \medskip

\noindent\textit{Mobilization}

\noindent The extent to which a platform can mobilize a social group depends
on two factors: the intensity of the group's belief that the platform induces
a good public outcome and the group's intrinsic attitude to the platform's
policy. Specifically, let $f_{i}:A\rightarrow%
\mathbb{R}
_{+}$ be a function that measures group $i$'s intrinsic attitude towards each
policy. Denote $f=(f_{i})_{i\in N}$. We refer to $f_{i}(a)$ as group $i$'s
\textquotedblleft\textit{mobilization potential}\textquotedblright\ for $a$.
Intuitively, $f_{i}(a)$ reflects the specific consequences (such as costs)
that policy~$a$ has for a group $i$ and so its raw support for~$a$, besides
its general interest in the public outcome.

Denote $N^{a}=\{i\in N\mid f_{i}(a)>0\}$---i.e., this is the set of groups
having positive mobilization potential for $a$. To avoid trivial cases, we
assume that for every group $i\in N$, $f_{i}(a)>0$ for least one policy $a$,
and that for every $a$, $f_{i}(a)>0$ for some $i\in N$.\medskip

\begin{definition}
[Platform Support]Fix a probability distribution $\sigma$ over platforms.
Then, group $i$'s support for a platform $(a,C,S)$ is defined by%
\begin{equation}
u_{i,\sigma}(a,C,S)=p_{\sigma}(y=1\mid x_{S}(a,C))\cdot f_{i}(a)
\label{platform support}%
\end{equation}

\end{definition}

Our notion of a social group's support for a platform takes a broad view of
political mobilization to include not only voting, but also other kinds of
political participation: rallies, petitions, protests, or social media
activism. We assume that the actual mobilization of group $i$ is
\textit{proportional} to its mobilization potential for the platform's policy,
as well as to the belief---shaped by the narrative---that the platform induces
a good outcome. The stronger a group's belief, the stronger its support for
the platform. In Section~\ref{subsec: microfoundation}, we provide a
\textquotedblleft microfoundation\textquotedblright\ that substantiates the
multiplicative form of $u_{i,\sigma}$.

The reason why $u_{i,\sigma}$ is indexed by $\sigma$ is that the belief
$p_{\sigma}(y=1\mid x_{S})$ may be sensitive to changes in $\sigma$. To see
why, recall that $y$ is purely a (probabilistic) function of $a$, so $y$ is
independent of $C$ conditional on $a$. This property can be represented by the
directed acyclic graph (DAG) $C\leftarrow a\rightarrow y$ (the direction of
the link between $C$ and $a$ is arbitrary; what matters is that they are
correlated because they are jointly determined, as we will describe below). If
the narrative $S$ does not attribute $y$ to $a$---i.e., $S\subseteq N$---it
amounts to interpreting the long-run correlation between $C$ and $y$ as if it
is causal, namely as if the DAG were $C\rightarrow y$. In reality, the
correlation is due to confounding, since $a$ is correlated with both $y$ and
$C$. The latter correlation is determined by $\sigma$ as the following
expression makes evident:%
\begin{equation}
p_{\sigma}(y=1\mid x_{S})=\sum_{a,C}p_{\sigma}(a,C\mid x_{S})p(y=1\mid a),
\label{p(y|X_S)}%
\end{equation}
where the term $p_{\sigma}(a,C\mid x_{S})$ is determined by $\sigma$.

To illustrate (\ref{p(y|X_S)}) and its role in how a false narrative can
induce a wrong belief, suppose $S=\{i\}$ and $i\notin C$ for some $i\in N$.
Then,%
\[
p_{\sigma}(y=1\mid x_{S}(\ell,C))=\frac{q\sum_{C,S\mid i\notin C}%
\sigma(h,C,S)}{\sum_{a,C,S\mid i\notin C}\sigma(a,C,S)}%
\]
If $\sigma(h,C,S)>0$ for some platforms in which $i\notin C$, then $p_{\sigma
}(y=1\mid x_{S}(\ell,C))>0$. This means that platform $(\ell,C,S)$ will
receive positive support (as long as $C$ includes some $j$ such that
$f_{j}(\ell)>0$), even though it objectively leads to $y=0$ with
certainty.\medskip

\noindent\textit{Admissible coalitions}

\noindent We impose a restriction on feasible coalitions, which is based on a
qualitative distinction between $f_{i}(a)>0$ and $f_{i}(a)=0$. Suppose group
$i$ is in fact intrinsically \textit{opposed} to policy $a$. Then, a coalition
that includes group~$i$ and advocates $a$ benefits from ousting that group
because it would act as a \textquotedblleft fifth column.\textquotedblright%
\ But of course group~$i$ would never join this coalition in the first place.
Therefore, we interpret $f_{i}(a)=0$ as an assumption that group~$i$ is
opposed to $a$ and, hence, will never be in a coalition that advocates it. By
assumption, the group satisfies $f_{i}(a^{\prime})>0$ for $a^{\prime}\neq a$,
which means that it could join coalitions that advocate $a^{\prime}$. In this
sense, rallying in favor of $a^{\prime}$ is like rallying against $a$.

Formally, given policy $a$, a coalition $C$ is\textit{ admissible} if
$f_{i}(a)>0$ for all $i\in C$. Hereafter, we consider only platforms involving
admissible coalitions.\medskip

\noindent\textit{Equilibrium}

\noindent The feedback between $\sigma$ and political mobilization impels us
to introduce an \textit{equilibrium} notion of dominant platforms.\medskip

\begin{definition}
[Equilibrium]\label{def: equilibrium} A probability distribution $\sigma$ with
full support over platforms with admissible coalitions is an $\varepsilon
$-equilibrium if whenever $\sigma(a,C,S)>\varepsilon$, $(a,C,S)$ maximizes%
\begin{equation}
U_{\sigma}(a,C,S)=\sum_{i\in C}u_{i,\sigma}(a,C,S) \label{eq}%
\end{equation}
A probability distribution $\sigma$ (not necessarily with full support) is an
equilibrium if it is the limit of $\varepsilon$-equilibria as $\varepsilon
\rightarrow0$.\medskip
\end{definition}

\noindent We start from the notion of $\varepsilon$-equilibrium to ensure that
$p_{\sigma}(y=1\mid x_{S})$ is always well-defined, given that groups rely on
long-run data to form beliefs. This \textquotedblleft trembling
hand\textquotedblright\ aspect plays a very limited role in our analysis.

Our equilibrium notion captures the idea that a platform's political strength
depends on how strongly it mobilizes its constituent groups. We view
narrative-fueled political competition as a battle over public opinion. A
platform is \textit{dominant} given $\sigma$ if it generates the largest
aggregate mobilization of its coalition members. Therefore, the only
\textit{dominant} platforms in equilibrium are those that maximize (\ref{eq}),
to which we refer as the \textit{platform's payoff}. When $(a,C,S)$ is
dominant, we say that $C$ is a \textit{ruling coalition}.

Note that if only the true narrative $S=\{0\}$ were feasible, then a platform
with $a=\ell$ would generate zero payoffs by (\ref{Pr(y|a)}) and the
definition of~$u_{i,\sigma}$. Instead, a platform with $a=h$ always generates
positive payoffs. In this case, $a=h$ would occur with probability one in
equilibrium. We therefore refer to $a=h$ as the \textquotedblleft%
\textit{rational}\textquotedblright\ policy.

\section{Two-Group Societies}

We begin our analysis with the simple case of $n=2$, where the only feasible
coalitions are $\{1\}$ and $\{2\}$. We assume that the mobilization potential
function satisfies $f_{1}(h)>f_{2}(h)$ and $f_{2}(\ell)>f_{1}(\ell)$. That is,
policy~$h$ ($\ell$) receives stronger raw support from group $1$ ($2$).

This specification of our model is akin to a two-party system, in which
exactly one party can be in power at any point in time. In this case, our
equilibrium concept can be interpreted in terms of a two-party voting model:
Supporters of each party vote non-strategically for it if their net
anticipatory utility from their party's platform is positive---otherwise they
abstain (somewhat as in \cite{razin2021misspecified}). We elaborate on this
connection in Section 6.1.

The following are a few real-life examples of outcomes and policies that we
have in mind. First, the issue is climate change and $a=h$ represents carbon
taxation, which produces a common environmental benefit but induces
differential costs among social groups (captured by $f$). Second, the issue is
economic growth, where $a=h$ represents policies such as deregulation or other
structural reforms that foster growth but inflict differential adjustment
costs across society. Finally, the issue is national security, where $a=h$
represents an aggressive military strategy that mitigates security threats but
involves sacrifices and moral judgments that vary across groups.

This simple setting allows us to reduce the set of relevant narratives. Since
$x_{1}=1$ if and only if $x_{2}=0$, all tribal narratives $S\subseteq N$ are
equivalent: They effectively say that \textquotedblleft\textit{things are good
when group }$i$\textit{ is in power / group }$j$\textit{ is not in
power}.\textquotedblright\ In addition, all narratives that weakly contain
$\{0\}$ are equivalent, because $\Pr(y=1\mid a,C)\equiv\Pr(y=1\mid a)$. Thus,
every feasible narrative is equivalent to one of the following three
narratives: the true narrative $\{0\}$, the denialist narrative $\varnothing$,
and the tribal narrative $\{1\}$ (or, equivalently, $\{2\}$). Therefore,
hereafter we assume that only these three narratives---which we denote by
$true$, $denial$, and $tribal$ for expositional clarity---are
feasible.\medskip

\begin{proposition}
\label{prop n=2}There is a unique equilibrium distribution over $(a,C)$. The
only platforms that can be in~$Supp(\sigma)$ are $(h,\{1\},true)$,
$(\ell,\{2\},denial)$, and $(\ell,\{1\},tribal)$. Furthermore:$\medskip
$\newline\noindent$(i)$ $\sigma(h,\{1\},true)=\min\left\{  1,f_{1}%
(h)/f_{2}(\ell)\right\}  $;$\medskip$\newline\noindent$(ii)$ $\sigma
(\ell,\{1\},tribal)>0$ only if $\sigma(\ell,\{2\},denial)>0$.$\medskip$
\end{proposition}

\noindent We present the proof of this result here, as it illustrates the
basic logic and forces at the core of our model. All other proofs are
relegated to the Appendix.$\medskip$

\paragraph{Proof of Proposition~\ref{prop n=2}}

\noindent We begin by writing the payoff $U_{\sigma}$ for the platforms
carried by the three relevant narratives:%
\begin{align*}
U_{\sigma}(a,\{i\},true)  &  =p_{\sigma}(y=1\mid a)\cdot f_{i}(a)\medskip\\
U_{\sigma}(a,\{i\},denial)  &  =p_{\sigma}(y=1)\cdot f_{i}(a)=q\cdot
p_{\sigma}(a=h)\cdot f_{i}(a)\medskip\\
U_{\sigma}(a,\{i\},tribal)  &  =p_{\sigma}(y=1\mid x_{i}=1)\cdot
f_{i}(a)\medskip
\end{align*}
The proof proceeds in steps.\medskip

\begin{step}
$(i)$ If $\sigma(a,\{i\},true)>0$, then $a=h$ and $i=1$. $(ii)$ If
$\sigma(h,\{i\},S)>0$, then $S=true$.
\end{step}

\begin{proof}
Consider any $\varepsilon$-equilibrium $\sigma$. Note that $p_{\sigma}(y=1\mid
a=\ell)=0$ and $p_{\sigma}(y=1\mid a=h)=q$. It follows that if $\sigma
(a,\{i\},true)>\varepsilon$ and hence $(a,\{i\},true)$ maximizes $U_{\sigma}$,
then $a=h$ and $i=1$ because $f_{1}(h)>f_{2}(h)$. Now suppose $\sigma
(h,\{i\},S)>\varepsilon$. Since $\sigma$ has full-support, $p_{\sigma}(y=1\mid
x_{S^{\prime}})<q$ whenever $0\notin S^{\prime}$. This means that $U_{\sigma
}(h,\{i\},true)>U_{\sigma}(h,\{i\},S^{\prime})$ for every such $S^{\prime}$;
hence, $S=true$. We have thus established that the claimed properties hold for
any $\varepsilon$-equilibrium and, hence, in any limit of $\varepsilon
$-equilibria.\medskip
\end{proof}

Step 1 implies that if $(a,\{i\},denial)$ or $(a,\{i\},tribal)$ are in
$Supp(\sigma)$, then $a=\ell$. \medskip

\begin{step}
If $\sigma(\ell,\{i\},denial)>0$, then $i=2$.
\end{step}

\begin{proof}
This follows immediately from $f_{2}(\ell)>f_{1}(\ell)$.\medskip
\end{proof}

\begin{step}
If $\sigma(\ell,\{i\},tribal)>0$, then $i=1$.
\end{step}

\begin{proof}
Step 1$(i)$ and $Pr(y=1|a=\ell)=0$ imply that $p_{\sigma}(y=1\mid x_{i}=1)>0$
only if $i=1$. Therefore, if $(a,\{i\},tribal)$ is in $Supp(\sigma)$, then
$i=1$.\medskip
\end{proof}

The previous steps pin down the three platforms that can be in $Supp(\sigma)$
for any equilibrium $\sigma$, as well as their payoffs:%
\begin{align}
U_{\sigma}(h,\{1\},true)  &  =q\cdot f_{1}(h)\medskip\label{n=2 eq equations}%
\\
U_{\sigma}(\ell,\{2\},denial)  &  =q\cdot\sigma(h,\{1\},true)\cdot f_{2}%
(\ell)\medskip\nonumber\\
U_{\sigma}(\ell,\{1\},tribal)  &  =q\cdot\frac{\sigma(h,\{1\},true)}%
{\sigma(h,\{1\},true)+\sigma(\ell,\{1\},tribal)}\cdot f_{1}(\ell
)\medskip\nonumber
\end{align}

\begin{step}
In any equilibrium, $\sigma(h,\{1\},true)>0$.
\end{step}

\begin{proof}
If $\sigma(h,\{1\},true)=0$, the previous expressions become
\[
U_{\sigma}(\ell,\{2\},denial)=U_{\sigma}(\ell,\{1\},tribal)=0<U_{\sigma
}(h,\{1\},true).
\]
Therefore, $(h,\{1\},true)$ is the only possible member of $Supp(\sigma)$, a
contradiction.\medskip
\end{proof}

\begin{step}
In any equilibrium, $\sigma(\ell,\{1\},tribal)>0$ only if $\sigma
(\ell,\{2\},denial)>0$.
\end{step}

\begin{proof}
Suppose $\sigma(\ell,\{1\},tribal)>0=\sigma(\ell,\{2\},denial)$. Then,%
\[
\sigma(h,\{1\},true)+\sigma(\ell,\{1\},tribal)=1,
\]
so that%
\[
U_{\sigma}(\ell,\{1\},tribal)=q\cdot\sigma(h,\{1\},true)\cdot f_{1}(\ell)
\]
But $f_{2}(\ell)>f_{1}(\ell)$ then implies that $U_{\sigma}(\ell
,\{1\},tribal)<U_{\sigma}(\ell,\{2\},denial)$, which contradicts $\sigma
(\ell,\{1\},tribal)>0$.\medskip
\end{proof}

We are now able to show that an equilibrium exists and it is unique. By Steps
1-3, there are three possible cases for $Supp(\sigma)$: all three platforms or
$\{(h,\{1\},true),(\ell,\{2\},denial)\}$ or $\{(h,\{1\},true)\}$. By the
definition of equilibrium, each member of $Supp(\sigma)$ has to maximize
$U_{\sigma}$.\medskip

\noindent\underline{Case I: $f_{1}(h)\geq f_{2}(\ell)$.} In this case,
$U_{\sigma}(h,\{1\},true)>U_{\sigma}(\ell,\{2\},denial)$ whenever
$\sigma(h,\{1\},true)<1$. Step 5 then implies $Supp(\sigma
)=\{(h,\{1\},true)\}$. Indeed, when $\sigma(h,\{1\},true)=1$,%
\[
U_{\sigma}(h,\{1\},true)\geq U_{\sigma}(\ell,\{2\},denial),U_{\sigma}%
(\ell,\{1\},tribal)
\]
Thus, $\sigma(h,\{1\},true)=1$ is the unique equilibrium.\medskip

\noindent\underline{Case II: $f_{2}(\ell)>f_{1}(h)$.} In this case,
$U_{\sigma}(h,\{1\},true)<U_{\sigma}(\ell,\{2\},denial)$ if $\sigma
(h,\{1\},true)=1$. Therefore, $\sigma(h,\{1\},true)<1$. Step 5 then implies
that $(\ell,\{2\},denial)\in Supp(\sigma)$. Now consider two sub-cases. First,
let $f_{2}(\ell)>f_{1}(h)\geq f_{1}(\ell)$. Then, $U_{\sigma}(\ell
,\{1\},tribal)<U_{\sigma}(h,\{1\},true)$ whenever $\sigma(\ell
,\{1\},tribal)>0$. Therefore,
\[
\sigma(h,\{1\},true)=\frac{f_{1}(h)}{f_{2}(\ell)}\qquad\sigma(\ell
,\{2\},denial)=\frac{f_{2}(\ell)-f_{1}(h)}{f_{2}(\ell)}%
\]
is the unique solution of%
\[
U_{\sigma}(\ell,\{2\},denial)=U_{\sigma}(h,\{1\},true)\geq U_{\sigma}%
(\ell,\{1\},tribal).
\]
Second, let $f_{2}(\ell)>f_{1}(\ell)>f_{1}(h)$. Then,%
\begin{align*}
\sigma(h,\{1\},true)  &  =\frac{f_{1}(h)}{f_{2}(\ell)}\\
\sigma(\ell,\{2\},denial)  &  =\frac{f_{2}(\ell)-f_{1}(\ell)}{f_{2}(\ell)}\\
\sigma(\ell,\{1\},tribal)  &  =\frac{f_{1}(\ell)-f_{1}(h)}{f_{2}(\ell)}%
\end{align*}
is the unique solution of%
\[
U_{\sigma}(\ell,\{2\},denial)=U_{\sigma}(h,\{1\},true)=U_{\sigma}%
(\ell,\{1\},tribal)\medskip
\]

\noindent This completes the proof. $\square$\bigskip

To interpret the equilibrium, consider the case in which
\[
f_{2}(\ell)>f_{1}(\ell)>f_{1}(h)>f_{2}(h)
\]
so that all three platforms described in the proposition are in $Supp(\sigma
)$. The distribution $\sigma$ describes the long-run frequencies with which
each of these platforms prevail. When $(h,\{1\},true)$ prevails, group $1$ in
power, implements the rational policy $h$, and employs the true narrative that
attributes outcomes to policies. When $(\ell,\{2\},denial)$ prevails, group
$2$ is in power, implements policy $\ell$, and employs the denial narrative
that implicitly attributes outcomes to external factors. Finally, when
$(\ell,\{1\},tribal)$ prevails, group $1$ is in power, implements $\ell$, and
employs the tribal narrative that attributes outcomes to who is in power
(without referring to policies).\medskip

\noindent\textit{A dynamic process}

\noindent For an intuition for this equilibrium, it is useful to have a
dynamic process in mind. Imagine that initially there are some random
fluctuations over $(a,C)$ and that only the true narrative is considered.
Then, policy~$\ell$ garners no support because it ensures $y=0$. The true
narrative can only justify policy $h$, which gets stronger support from group
$1$. Thus, the prevailing platform is $(h,\{1\},true)$, and it generates a
payoff of $q\cdot f_{1}(h)$.

Now suppose that at some point, the platform $(\ell,\{2\},denial)$ arises.
Since approximately only policy~$h$ that has ever been taken, the denialist
narrative induces the belief $\Pr(y=1)\approx q$. Because $f_{2}(\ell
)>f_{1}(h)$, the new platform generates stronger support than the
\textquotedblleft incumbent\textquotedblright\ platform $(h,\{1\},true)$. As a
result, the new platform becomes dominant, displacing the old one. Since the
new platform involves $a=\ell$, the historical frequency of $a=h$ gradually
declines, lowering $\Pr(y=1)$.

As this process continues, the denialist platform's payoff will drop below
$q\cdot f_{1}(\ell)$. When this happens, a third narrative can gain traction
and give rise to yet another platform $(\ell,\{1\},tribal)$. In the path
described above, $a=h$ if and only if $x_{1}=1$ (approximately). This implies
the historical conditional probability $\Pr(y=1\mid x_{1}=1)\approx q$.
Consequently, a narrative arguing that things are good when group $1$ is in
power (or, equivalently, when group $2$ is out of power) can mobilize group
$1$ behind policy $\ell$. The payoff of $(\ell,\{1\},tribal)$ is approximately
$q\cdot f_{1}(\ell)$. Since $f_{1}(\ell)>f_{1}(h)$, this payoff exceeds that
of the two previous platforms and $(\ell,\{1\},tribal)$ becomes dominant.
While it dominates, it weakens the correlation between $x_{1}$ and $y$ and
therefore lowers its own payoff. It also weakens the appeal of the denial
narrative by lowering the frequency of $y=1$. This brings the platform carried
by the true narrative back in vogue.%

\[%
{\parbox[b]{3.8261in}{\begin{center}
\includegraphics[
height=2.8721in,
width=3.8261in
]%
{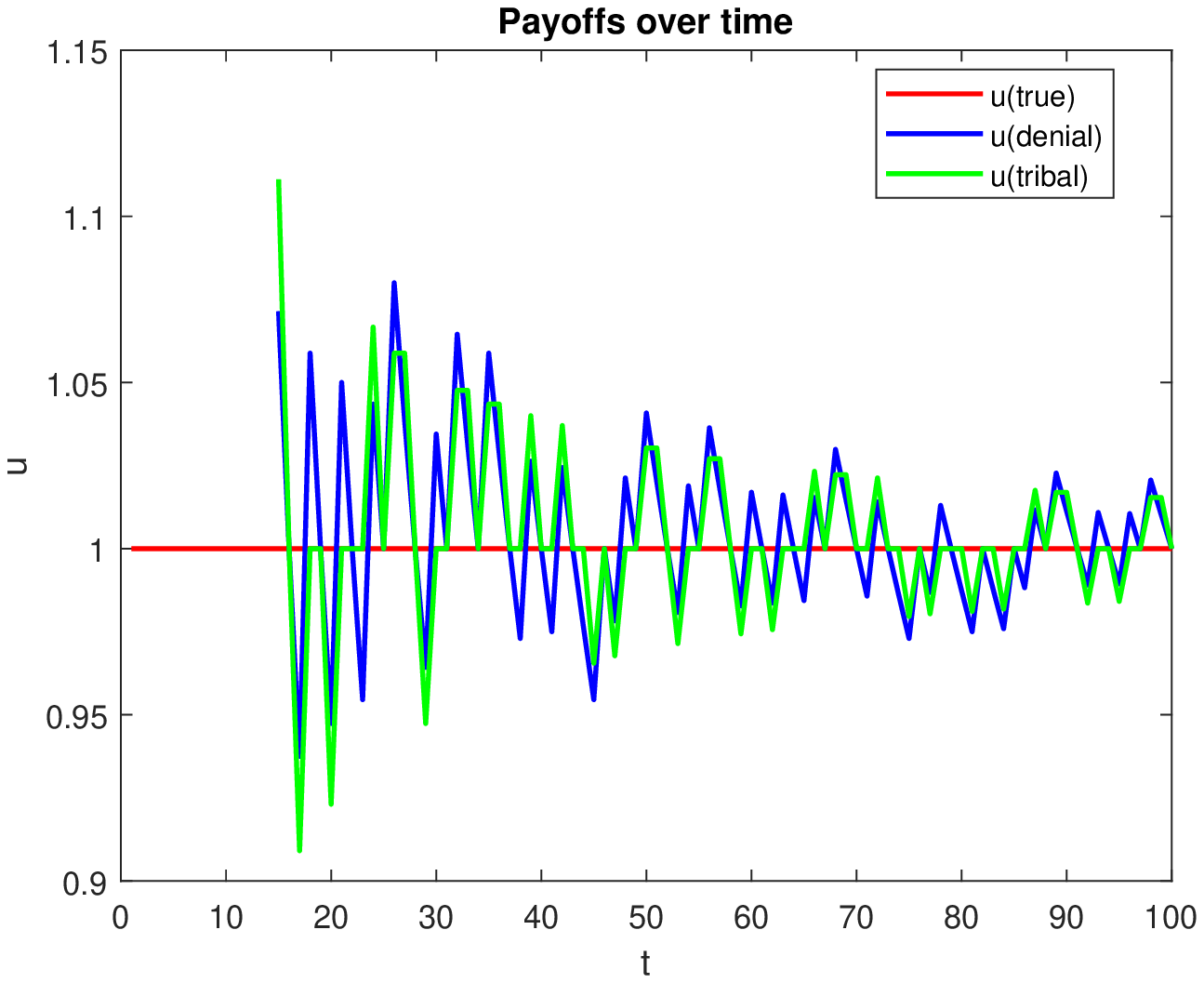}%
\\
Figure 1
\end{center}}}
\]

The subsequent dynamic repeats this cycle, albeit with smaller payoff swings.
In the long run, the process reaches a configuration in which all three
platforms generate the same payoff $q\cdot f_{1}(h)$. From this point on, any
deviation that increases the long-run frequency of one of the three platforms
will trigger an offsetting dynamic response. That is, the equilibrium is
dynamically stable. Figure 1 shows the evolution of the three platforms'
payoffs, assuming $f_{2}(\ell)=3$, $f_{1}(\ell)=2$, and $f_{1}(h)=1$.

We wish to highlight a\ few significant features of Proposition~\ref{prop n=2}%
. The rational policy must be played with positive probability in equilibrium.
The reason is that any platform carried by a false narrative free-rides on
episodes with $a=h$. Also, note that a platform that advocates $a=h$ will
generate its largest support if it employs the true narrative, which
highlights the correlation between $a$ and $y$ (since this correlation is
stronger than the correlation between $y$ and any other variable).

However, when $f_{2}(\ell)>f_{1}(h)$, policy $\ell$ has stronger raw support
than the rational policy $h$. In this case, false narratives allow $\ell$ to
gain dominance at the expense of $h$. They enable supporters to
\textquotedblleft eat their cake and have it:\textquotedblright\ On the one
hand, they are attracted to the policy; on the other hand, the narrative
distracts them from the adverse objective consequences of this policy. The
equilibrium probability of $a=h$ is determined by the ratio $f_{1}%
(h)/f_{2}(\ell)$. What makes policy $\ell$ not only popular but also
\textquotedblleft populist\textquotedblright\ is that it necessitates a false
narrative to mobilize public opinion.

The distinction between the two false narratives---i.e., the denial and tribal
narratives---is irrelevant for the equilibrium probability of $a=h$. However,
it does play an important role for the identity of the group in power. When
$f_{1}(\ell)>f_{1}(h)$, the tribal narrative enables group $1$ to displace
group~$2$, even though it adopts the same \textquotedblleft
populist\textquotedblright\ policy $\ell$. The reason is that group $1$ can
milk its reputation for achieving a good outcome---thanks to its actual
historical tendency to implement $a=h$, which is greater than its rival's. It
does so by highlighting the historical correlation between $y=1$ and being in
power (or, equivalently, group $2$ being out of power). In the next section,
we will see that in societies with more than two groups, tribal narratives are
relevant not only for the power structure, but also for the equilibrium
frequency of the rational policy.

\section{Fragmented Societies}

\label{sec: fragmented_societies}

In this section we characterize equilibria in the presence of more than two
social groups (i.e., $n>2$). This creates a distinction between varieties of
tribal narratives absent in the case of $n=2$: There will be a difference
between \textquotedblleft exclusionary\textquotedblright\ narratives of the
form \textquotedblleft things are good when these groups are \textit{out }of
power\textquotedblright\ and \textquotedblleft inclusionary\textquotedblright%
\ narratives of the form \textquotedblleft things are good when these groups
are \textit{in }power.\textquotedblright\ We will see how proliferation of
exclusionary narratives can have a detrimental effect on the equilibrium
probability of the rational policy. It also leads to new coalitional
structures that would not arise if only the true narrative prevailed.

Toward this end, we impose some structure on the set of feasible narratives.
Let\ $\mathcal{S}$ be a family of subsets of $N$. A narrative $S$ is feasible
if and only if $S\setminus\{0\}\in\mathcal{S}$. We interpret each element in
$\mathcal{S}$ as a collection of social groups that can be \textit{clearly
identified} by a common label or defining attribute. For instance,
$\mathcal{S}$ can represent a division of society along political leanings or
ethnic-religious affiliations, or it can represent salient collections of
social groups (such as \textquotedblleft unionized workers\textquotedblright%
\ or \textquotedblleft the economic elite\textquotedblright).

The motivation for restricting the domain of feasible narratives is that,
depending on the context, not every collection of social groups can be
identified and held accountable for outcomes, because it may lack a clear
representation in governing institutions. In some political systems (e.g.,
Israel) there are political parties that directly represent specific
ethno-religious groups. Consequently, there is data about their power status
and how it is correlated with outcomes, thus making a narrative that exploits
this correlation quantifiable. In other systems (e.g., the US), the mapping
between social groups and political representation is more blurred, thus
restricting the supply of similar narratives.

Recall that $N^{a}=\{i\in N\mid f_{i}(a)>0\}$. Henceforth, we assume that the
sets $N\setminus N^{h}$ (groups that support only $\ell$), $N\setminus
N^{\ell}$ (groups that support only $h$), and $N^{\ell}\cap N^{h}$ (groups
that support both policies) are non-empty and that they are all in
$\mathcal{S}$. We sometimes refer to these broad categories as
\textquotedblleft right,\textquotedblright\ \textquotedblleft
left,\textquotedblright\ and \textquotedblleft center.\textquotedblright%
\ These categories\ are always feasible as tribal narratives. We also assume
that any other $S\in\mathcal{S}$ is a \textit{subset} of one of these three
categories. Finally, we assume that the \textit{true and denialist narratives
are always feasible}.\medskip

\noindent\textit{An illustrative example}

\noindent Let $n=4$ and $\mathcal{S}=\{\{1\},\{2\},\{3\},\{4\},\{3,4\}\}$. The
left\ is $\{1\}$, the center\ is $\{2\}$, and the right\ is $\{3,4\}$%
;\ $\{3\}$ and $\{4\}$ represent sub-divisions of the right (e.g., moderates
and extremists). Let $f_{3}\equiv f_{4}$. Assume further that $f_{2}%
(\ell)>f_{1}(h)+f_{2}(h)$, namely the center's raw support for $\ell$ is
stronger than the raw support for $h$ among the center-left. The following
distribution is an equilibrium (we will shortly see that it is the
\textquotedblleft essentially unique\textquotedblright\ equilibrium):%
\[%
\begin{array}
[c]{cccc}%
\sigma & policy & coalition & narrative\\
\frac{f_{1}(h)+f_{2}(h)}{f_{2}(\ell)+f_{3}(\ell)+f_{4}(\ell)}\medskip & h &
\{1,2\} & true\\
\frac{f_{2}(\ell)-f_{1}(h)-f_{2}(h)}{f_{2}(\ell)+f_{3}(\ell)+f_{4}(\ell
)}\medskip & \ell & \{2\} & \{3,4\}\\
\frac{f_{3}(\ell)+f_{4}(\ell)}{2[f_{2}(\ell)+f_{3}(\ell)+f_{4}(\ell)]}\medskip
& \ell & \{2,3\} & \{4\}\\
ditto & \ell & \{2,4\} & \{3\}
\end{array}
\]

As in the previous section, policy $\ell$ is played with positive probability,
sustained by false narratives. In contrast to that case, however, here all
false narratives take the exclusionary tribal\ form. For example, in the
platform $(\ell,\{2\},\{3,4\})$, the center attributes a good outcome to
keeping the right out of power. Furthermore, the equilibrium exhibits
\textit{endogenous fragmentation}: The center and each faction of the right
sometimes form a coalition, using a false narrative that attributes the good
outcome to keeping the remaining right-wing group out of power.

Finally, note that as in the previous section, the equilibrium probability of
the rational policy $h$ is equal to the ratio between the total mobilization
potential for $h$ among the groups that implement $h$ in equilibrium (i.e.,
center-left) and the total mobilization potential for $\ell$ among the groups
that implement $\ell$ in equilibrium (i.e., center-right). We will see in the
next sub-section that this is not a general feature. \hfill$\square$\bigskip

We now establish existence of equilibrium and show that the joint equilibrium
distribution over policies and coalitions is unique. The proof provides an
algorithm for computing this distribution. Since multiple narratives can
induce the same beliefs, we introduce a natural refinement that pins down the
equilibrium distribution over platforms.\medskip

\begin{definition}
[Essential Equilibrium]\label{def: essential equilibrium} An equilibrium
$\sigma$ is essential if it satisfies two properties:

\begin{itemize}
\item[$(i)$] If $(a,C,S)\in Supp(\sigma)$ and $S\neq\{0\}$, then $U_{\sigma
}(a,C,\{0\})<U_{\sigma}(a,C,S)$.

\item[$(ii)$] If $(a,C,S)\in Supp(\sigma)$ and $S\subseteq N$, then
$U_{\sigma}(a,C,S^{\prime})<U_{\sigma}(a,C,S)$ for all $S^{\prime}\subset
S$.\medskip
\end{itemize}
\end{definition}

\noindent This definition means that if several narratives that accompany
$(a,C)$ induce the same payoff, there is a lexicographically secondary
criterion that favors \textit{true over false} narratives (condition $(i)$)
and \textit{simple over complex} narratives (condition $(ii)$). We introduce
refinement $(i)$ in order to isolate the cases where false narratives are
necessary for the platform's dominance. We introduce refinement $(ii)$ in
order to isolate the components of false narratives that are essential for the
belief they generate.

Our main result uses the following piece of notation:%
\[
F(a,M)=\sum_{i\in M}f_{i}(a),\qquad M\subseteq N.
\]
That is, $F(a,M)$ is the aggregate mobilization potential for policy $a$ in
the collection of groups $M$. Recall that $N^{a}$ is the set of groups with
positive mobilization potential for $a$.\medskip

\begin{proposition}
\label{prop: gen_charact_one_issue} There exists a unique essential
equilibrium $\sigma^{\ast}$. This equilibrium satisfies the following properties:

\begin{itemize}
\item[$(i)$] The rational policy $a=h$ is realized only as part of the
platform $(h,N^{h},\{0\})$.

\item[$(ii)$] The probability of $a=h$ is $1$ if $F(h,N)\geq F(\ell,N)$, and
belongs to $(0,F(h,N)/F(\ell,N)]$ otherwise.

\item[$(iii)$] Every platform $(\ell,C,S)\in Supp(\sigma^{\ast})$ satisfies
$S\subseteq N\setminus N^{h}$ and $C=N^{\ell}\setminus S$.\medskip
\end{itemize}
\end{proposition}

\begin{remark}
While Proposition \ref{prop: gen_charact_one_issue} focuses on essential
equilibrium, its proof effectively establishes uniqueness of the equilibrium
distribution over $(a,C)$. The restriction to essential equilibria serves to
pin down the essential components of the narratives that accompany each
policy-coalition pair.\bigskip
\end{remark}

The unique essential equilibrium has the following noteworthy features. First,
the rational policy $h$ is always taken with positive probability. Second,
this is the only policy taken in equilibrium if and only if it has higher
aggregate intrinsic support than $\ell$ (i.e., $F(h,N)\geq F(\ell,N)$). Third,
when policy $\ell$ is taken with positive probability, it is accompanied by
tribal false narratives that take the exclusionary form: They identify a
collection $S\in\mathcal{S}$ of social groups that oppose $h$ but are not part
of the coalition supporting $\ell$, and essentially argue that
\textquotedblleft things are good when $S$ is out of power.\textquotedblright%
\ When $S=\varnothing$, this is reduced to the denialist narrative.

In the case of $n=2$, exclusionary and inclusionary tribal narratives were
equivalent. This is no longer the case when $n>2$. What makes exclusionary
tribal narratives effective in this context? When a social group opposes $h$,
there is positive correlation between the good outcome and the group being out
of power. The exclusionary tribal narrative allows a coalition to exploit this
correlation to generate a false belief among its members that the very
exclusion of some groups from the coalition will lead to a good outcome, while
advocating policy~$\ell$. This device enables the coalition to
\textquotedblleft\textit{have its cake and eat it}:\textquotedblright\ The
coalition reaps the mobilization benefits of the intrinsically more attractive
policy $\ell$, while using the tribal narrative to deflect responsibility for
a bad outcome away from this policy and instead ``scapegoat'' the excluded
groups for it.

Exclusionary tribal narratives balance a trade-off between breadth and
intensity of the support they generate. Excluding groups from a coalition is
costly because it forgoes their support and thus lowers its aggregate
mobilization potential. However, if this exclusion is not too frequent, its
correlation with $a=h$ (and hence $y=1$) remains sufficiently strong, thus
generating intense support from the coalition's members. At one extreme, the
denialist narrative garners the largest coalition because it does not exclude
any group, but this comes at the cost of a weaker belief of $y=1$ and
therefore weaker mobilization of coalition members.

In fact, the denialist narrative plays an important role in the equilibrium
characterization. The upper bound on $\sigma^{\ast}(h,N^{h},\{0\})$ in
Proposition \ref{prop: gen_charact_one_issue} is the probability of $a=h$ in a
model in which the only feasible narratives are the true and denialist ones.
To see why, note that in such a model the payoff from platform $(h,N^{h}%
,\{0\})$ is $q\cdot F(h,N)$, whereas the payoff from $(\ell,N^{\ell
},\varnothing)$ is $q\cdot\sigma^{\ast}(h,N^{h},\{0\})\cdot F(\ell,N)$. Since
$(h,N^{h},\{0\})$ is always in the support of an essential equilibrium, the
payoff from $(\ell,N^{\ell},\varnothing)$ cannot exceed the payoff from
$(h,N^{h},\{0\})$. This implies the upper bound on\ $\sigma^{\ast}%
(h,N^{h},\{0\})$.

Why do equilibrium platforms advocating policy $\ell$ refrain from employing
\textquotedblleft inclusionary\textquotedblright\ tribal narratives $S$ that
attribute the outcome to members of the platform's coalition? The answer is
that in order to generate a positive belief, $S$ must also be contained in the
coalition that supports policy $h$. But since the inclusion of $S$ in a ruling
coalition is associated with a good outcome, $S$ cannot be part of any
\textquotedblleft exclusionary\textquotedblright\ tribal narrative that
accompanies policy $\ell$. This means that $S$ will be part of every coalition
that supports $\ell$. It follows that $S$ will $always$ be in power. As a
result, its power status is uncorrelated with $y$ and therefore a redundant
aspect of any narrative that cannot add to its explanatory power. Part $(ii)$
of the definition of essential equilibrium thus rules out the use of
inclusionary tribal narratives.

\section{Two Special Cases}

In this section, we provide a richer characterization for two specifications
of~$\mathcal{S}$. Throughout this section, we assume $F(\ell,N^{h}%
)>F(h,N^{h})$---i.e., among the groups that support $h$, there is greater
aggregate mobilization potential for $\ell$. This means that while $h$ is the
better policy for $y$, $\ell$ has larger intrinsic appeal.

\subsection{Narratives Based on a Social Taxonomy}

In this sub-section, we assume that $\mathcal{S}$ takes the form of a sequence
of progressively finer \textit{partitions} $\Pi=(\pi_{1},...,\pi_{K})$. That
is, for every $k=2,...,K$, $\pi_{k}$ is a partition of $N$ that is a
refinement of the partition $\pi_{k-1}$. Note that by the general restrictions
we imposed on the domain of feasible narratives, $\pi_{1}$ consists of the
non-empty sets $N^{\ell}\cap N^{h}$, $N\setminus N^{h}$ and $N\setminus
N^{\ell}$. We refer to $\Pi$ as a \textit{social taxonomy}. We sometimes abuse
notation and write $S\in\Pi$ to refer to a cell in one of the partitions in
$\Pi$.

We interpret each cell in a partition as a collection of social groups that
has a clear label or a defining attribute. For example, the broadest partition
can consist of the left, the right, and the center, with finer partitions
representing finer distinctions (e.g., \textquotedblleft
progressives\textquotedblright\ and \textquotedblleft
moderates\textquotedblright\ within the Democratic wing in US politics).
Another example is the division of Israeli society into increasingly fine
ethnic-religious groups (e.g., \textquotedblleft secular
Jewish\textquotedblright\ vs.\thinspace\textquotedblleft religious
Jewish\textquotedblright\ at a coarse level, then splitting the latter group
into \textquotedblleft nationalist-religious\textquotedblright\ and
\textquotedblleft ultra-orthodox\textquotedblright\ sub-categories). Finally,
a typical socioeconomic division of Italian society is between employees,
retirees, and entrepreneurs. Employees are further divided into those who are
unionized and highly protected and those who are not.

We impose some structure on $\Pi$: For every layer $k>1$ there is an integer
$r_{k}>1$ such that, for every $S\in\pi_{k-1}$, there are \textit{exactly}
$r_{k}$ sets $S^{\prime}\in\pi_{k}$ that satisfy $S^{\prime}\subset S$. In
other words, at every layer $k-1<K$ in the social taxonomy, each category is
split into the same number $r_{k}$ of sub-categories in the next layer. For
$k=1$, let $r_{1}=1$. Denote%
\[
R=\sum_{k=1}^{K}(r_{k}-1)
\]
This number is a natural measure of how fragmented the social taxonomy is,
because it increases with the number of levels $K$ and with the number of
splits-per-cell $r_{k}$ at each layer $k$ of the social taxonomy.\medskip

\begin{proposition}
\label{prop: eq_categories} The unique essential equilibrium $\sigma^{\ast}$
satisfies%
\begin{equation}
\sigma^{\ast}(h,N^{h},\{0\})=\frac{F(h,N)}{F(\ell,N^{h})+\max(R,1)\cdot
F(\ell,N\setminus N^{h})} \label{formula taxonomy}%
\end{equation}
Moreover:\medskip\newline$(i)$ Every $S\subseteq N\setminus N^{h}$ that
belongs to one of the partitions in $\Pi$ is used as a narrative in
$\sigma^{\ast}$.\medskip\newline$(ii)$ The denialist narrative is used in the
support of $\sigma^{\ast}$ if and only if $K=1$.\medskip
\end{proposition}

Let us highlight a few features of the unique essential equilibrium in the
social-taxonomy case. First, we obtain an explicit formula for the probability
that the rational policy is taken, which decreases with $R$. Thus, political
fragmentation can be detrimental to the implementation of socially beneficial
policies because it creates more room for false tribal narratives.

Second, every social category that is weakly finer than $N\setminus N^{h}$
serves as an exclusionary tribal narrative in the equilibrium. To see the
intuition, suppose that some category in the social taxonomy is invoked by an
exclusionary tribal narrative, but one of its direct sub-categories is never
used. That is, the support of $\sigma$ includes a platform that uses some
narrative $S\in\pi_{k}$, but there is a narrative $S^{\prime}\subset S$ that
belongs to $\pi_{k+1}$ and is not employed by any platform. Because of the
nested partition structure of $\Pi$, the cells that weakly contain $S$ and
$S^{\prime}$ are the same. That is, realization of $\sigma^{\ast}$ in which
the members of $S$ are not in power are the same as the realizations in which
the members of~$S^{\prime}$ are not in power. This means that the narratives
$S$ and $S^{\prime}$ generate the same beliefs. However, the larger coalition
that attributes the outcome to $S^{\prime}$ would generate stronger
mobilization. Therefore, a platform that employs the narrative $S^{\prime}$
would have a strictly higher payoff, which cannot happen in equilibrium.

Third, the denialist narrative is used only in the special case in which the
social taxonomy lacks a finer distinction than the broad classification given
by $\pi_{1}$. In this case, $\sigma^{\ast}(h,N^{h},\{0\})$ is equal to the
upper bound $F(h,N)/F(\ell,N)$. The only other case in which the latter
property holds is when the social taxonomy breaks $N\setminus N^{h}$ into
precisely two sub-categories, and there is no finer categorization (in this
case, $R=1$).

Formula (\ref{formula taxonomy}) enables us to subject $\sigma^{\ast}%
(h,N^{h},\{0\})$ to some simple comparative statics with respect to the
primitives of the model. First, a \textit{richer taxonomy} (as captured by
increasing $R$) lowers the equilibrium value of $\sigma^{\ast}(h,N^{h}%
,\{0\})$, due to the proliferation of false narratives in equilibrium. In this
sense, a more fragmented political system intensifies the \textquotedblleft
tribal\textquotedblright\ character of public opinion, and this has a cost in
terms of the common good.

Now consider changes in the mobilization potential function that reflect
\textit{more polarized attitudes} toward policy $\ell$. Specifically, suppose
that we lower $F(\ell,N^{h})$ and raise $F(\ell,N\setminus N^{h})$, while
keeping $F(\ell,N)$ fixed (without upsetting the inequality\ $F(\ell
,N^{h})>F(h,N^{h})$ that we assumed at the outset). This captures a shift of
support for $\ell$ from the \textquotedblleft center\textquotedblright\ to the
\textquotedblleft right\textquotedblright, resulting in a more polarized
society. When the social taxonomy satisfies $R>1$, this shift lowers
$\sigma^{\ast}(h,N^{h},\{0\})$. In this sense, greater polarization is
detrimental for the common good.

\subsection{A Rich Domain of Narratives}

In this sub-section, we consider the extreme case in which $\mathcal{S}$ is
the set of \textit{all} subsets $S\subseteq N$ that are weakly contained in
$N\setminus N^{h}$, $N\setminus N^{\ell}$, or $N^{\ell}\cap N^{h}$. We refer
to this specification of $\mathcal{S}$ as the \textquotedblleft
rich\textquotedblright\ narrative domain.\medskip

\begin{proposition}
\label{prop: eq_rich_domain} Under the rich narrative domain, the unique
essential equilibrium $\sigma^{\ast}$ satisfies the following
properties:\medskip\newline$(i)$ $\sigma^{\ast}(h,N^{h},\{0\})=\frac
{F(h,N)}{F(\ell,N)}$\medskip\newline$(ii)$ The rest of\ $Supp(\sigma^{\ast})$
consists of all platforms $(\ell,N^{\ell}\setminus S,S)$ such that $S\subseteq
N\setminus N^{h}$ and $\mid S\mid\geq\mid N\setminus N^{h}\mid-1$.\medskip
\end{proposition}

Thus, when the narrative domain is rich, the equilibrium structure is simple.
The rational policy is played with the maximal possible probability, as given
by our main result. As usual, the rational policy is carried by the true
narrative and supported by the coalition $N^{h}$. All false narratives are
tribal. In each of them, the collection of groups that is held accountable is
either $N\setminus N^{h}$ (i.e., everyone that opposes $h$) or a subset that
excludes exactly one group from this set. In these cases, this excluded group
joins the coalition supporting $\ell$.

We can see that even though the domain of possible tribal narratives is
maximally rich, only a small subset of them is used. The result also
demonstrates the non-monotonicity of the effect of tribalism and political
fragmentation on the equilibrium probability of the rational policy. The rich
domain represents a larger scope for tribal narratives than the taxonomies
studied in the previous sub-section. Nevertheless, the probability of rational
behavior is higher than for most social taxonomies. The reason is that if a
platform uses an exclusionary tribal narrative that is given by a small set
$S$, the collection of narratives in $\mathcal{S}$ that contain $S$ is larger
than in the taxonomy case. This means that in the rich domain case, using
small narratives $S$ significantly dilutes the belief that $y=1$, because
there are many instances where $S$ is not in power and $a=\ell$. This effect
deters the entry of platforms that rely on such small narratives, and
therefore blocks their proliferation.

\section{Model Foundations}

\subsection{A \textquotedblleft Microfoundation\textquotedblright\ for
Mobilization Potentials}

\label{subsec: microfoundation}

In this section, we illustrate how our definition of group $i$'s support for
platforms, $u_{i,\sigma}$, can be derived from a more detailed model in which
\textit{anticipatory utility} drives political participation.

Each social group $i$ consists of a measure $m_{i}$ of individuals. Each
individual has a payoff function $v(y,a)=y-c_{a,i}$, where the idiosyncratic
cost $c_{a,i}$ is drawn from the uniform distribution over $[0,\bar{c}_{a,i}%
]$. Thus, the constant $\bar{c}_{a,i}>0$ is an action-specific characteristic
of the social group $i$.

Suppose group $i$ belongs to a coalition $C$ that is part of the platform
$(a,C,S)$. The group's conditional belief over $y$ is given by $p_{\sigma
}(y\mid x_{S}(a,C))$. The platform mobilizes an individual in group $i$ with
cost $c_{a,i}$ if and only if his subjective anticipatory utility from the
platform is positive---that is, if $p_{\sigma}(y=1\mid x_{S}(a,C))>c_{a,i}$.
As a result, the total mass of mobilized individuals from group $i$ is%
\[
m_{i}\cdot\frac{p_{\sigma}(y=1\mid x_{S}(a,C))}{\bar{c}_{a,i}}%
\]
This is consistent with our definition of $u_{i,\sigma}$ with respect to a
mobilization potential $f_{i}(a)=m_{i}/\bar{c}_{a,i}$.

It should be emphasized that this \textquotedblleft
microfoundation\textquotedblright\ by no means rationalizes the individuals'
behavior, because they are not engaged in consequentialist analysis of their
participation decision. Nevertheless, it offers a deeper understanding of the
psychology of motivated reasoning behind our mobilization-potential function.

\subsection{A Dynamic Convergence Result}

In this section, we consider a simple and natural dynamic process that
determines which platforms garner maximal support over time. We show that the
process converges to the unique equilibrium distribution over policies and
coalitions in our main result. This global convergence result provides a
dynamic foundation for our equilibrium concept.

Time is discrete and denoted by $t=1,2,\ldots$. In each period $t$, there is a
distribution $\sigma_{t}$ over platforms $(a,C,S)$, where $a\in\{\ell,h\}$,
$C\subseteq N$, and $S\in\mathcal{S}$. Let the initial $\sigma_{1}$ be any
distribution with full support over the set of platforms using admissible
coalitions. Since the set of platforms is finite, this distribution is
well-defined. The distribution $\sigma_{t}$ evolves according to the following
adjustment. For every $t\geq2$, let
\[
\overline{(a,C,S)}_{t}\in\underset{(a^{\prime},C^{\prime},S^{\prime})}%
{\arg\max}\,U_{\sigma_{t}}(a^{\prime},C^{\prime},S^{\prime}),
\]
where ties can be broken arbitrarily. Then, let
\[
\sigma_{t+1}(a,C,S)=%
\begin{cases}
\frac{1}{t+1}+\frac{t}{t+1}\sigma_{t}(a,C,S)\qquad & \text{if }%
(a,C,S)=\overline{(a,C,S)}_{t}\\
\, & \\
\frac{t}{t+1}\sigma_{t}(a,C,S)\qquad & \text{otherwise.}%
\end{cases}
\]
Thus, for $t$ large enough, we can essentially view $\sigma_{t}(a,S,C)$ as the
empirical frequency with which platform $(a,C,S)$ has been dominant in the
available history of data.\medskip

\begin{proposition}
\label{prop: dynamics} Every limit point $\sigma$ of the process $\sigma_{t}$
induces the same distribution over policy-coalition pairs $(a,C)$ as that
induced by the unique essential equilibrium $\sigma^{\ast}$.\medskip
\end{proposition}

This result formalizes and generalizes the dynamic convergence process we
discussed in the context of the two-group specification in Section 3.\medskip

\noindent\textit{Comment: An alternative dynamic interpretation}

\noindent Throughout the paper, we interpreted the platform support function
$u_{i,\sigma}(a,C,S)$ given by (\ref{platform support}) as if agents have a
long memory, which enables them to measure the long-run conditional frequency
$p_{\sigma}(y=1\mid x_{S}(a,C))$. This interpretation also underlies the
dynamic process considered in this sub-section.

An alternative interpretation is that agents reason anecdotally. When a
combination $a,C$ is proposed and a narrative $S$ is suggested at some time
period, agents are reminded of the last period in which $x_{S}(a,C)$ took the
same value. For example, if the narrative attributes the outcome to whether
group $i$ is in power, agents' attention is drawn to the last period in which
this group was indeed in power, and they record the realization of $y$ during
that period. This means that $p_{\sigma}(y=1\mid x_{S}(a,C))$ does not
represent a probabilistic belief, but rather the probability that a random
anecdote will be favorable for the platform $(a,C,S)$.

This alternative interpretation is realistic in the sense that it reflects
voters' short memory and tendency to overreact to single, recent anecdotes.
Whether a dynamic process based on this interpretation induces long-run
frequencies of dominant platforms that match our notion of equilibrium is left
as an open problem.

\section{Related Literature}

\cite{ES2020} presented the basic idea of formalizing political narratives as
causal models whose adoption by agents is driven by motivated reasoning. The
present paper borrows these ingredients. The key difference is that
\cite{ES2020} considered a representative-agent model, whereas the present
paper assumes a heterogeneous society and focuses on the role of false
narratives as the \textquotedblleft glue\textquotedblright\ of social
coalitions. Indeed, the existence of multiple social groups in the present
model enables the new class of \textquotedblleft tribal\textquotedblright%
\ narratives, which are moot in the single-agent model.

More broadly, this paper is related to a strand in the political-economics
literature that studies voters' belief formation according to misspecified
subjective models or wrong causal attribution rules (e.g., \cite{Spiegler2013}%
, \cite{EP2017}, \cite{izzo2021ideological}, and \cite{razin2021misspecified}%
). In particular, the latter paper studies dynamic electoral competition
between two candidates, each associated with a different subjective model of
how two policy variables map into outcomes. One model is complete and correct;
the other is a \textquotedblleft simplistic\textquotedblright\ model that
omits one of the policy variables. Voter participation is costly; stronger
beliefs lead to larger voter turnout. The long-run behavior of this system
involves ebbs and flows in the relative popularity of the two models.

The general program of studying the behavioral implications of misspecified
causal models is due to Spiegler (\citeyear{Spiegler2016,Spiegler2020}). In
their general form, causal models are formalized as directed acyclic graphs,
following the Statistics/AI literature on graphical probabilistic models
(\cite{CLDSNLJ1999}, \cite{Pearl2009}). The causal models in this paper fit
into the graphical formalism (as we saw in Section 2), but do not require its
heavy use because they take a relatively simple form (this form is related to
misspecified models in otherwise very different works, such as
\cite{Jehiel2005}, \cite{EP2013} or \cite{MS2020}). Therefore, in this paper,
graphical representations of causal models remained mostly in the background.

Given the fluidity of the notion of narratives, it naturally invites diverse
formalizations. \cite{benabou2018narratives} focus on moral decision-making
and formalize narratives as messages or signals that can affect
decision-makers' beliefs regarding the externality of their policies.
\cite{levy2021maximum} use the term to describe information structures in
game-theoretic settings that people postulate in order to explain observed
behavior. Schwartzstein and Sunderam
(\citeyear{schwartzstein2021shared,SS2021}) propose an alternative approach to
\textquotedblleft persuasion by models,\textquotedblright\ where models are
formalized as likelihood functions and the criterion for selecting models is
their success in accounting for historical observations. \cite{Shiller2017}
focuses on the spread of economic narratives in society, using an
epidemiological analogy.

The political science literature has long acknowledged the power of narratives
in garnering public support for policies and in mobilizing people to protests
or rallies (see \cite{polletta2008storytelling}). In particular, the so-called
\textquotedblleft narrative policy framework\textquotedblright\ was developed
as a systematic empirical framework for studying the role of stories or
narratives in public policy. Studies employing this framework have argued that
narratives have a greater influence on the opinions of policymakers and
citizens than does scientific information (see, e.g., \cite{SJM2011},
\cite{JM2010}, and \cite{JMS2014}).

Finally, there are a few recent attempts to study political and economic
narratives empirically, using textual analysis. Mobilizing public opinion
often takes the form of texts (speeches, op-eds, tweets). What we observe in
these texts are qualitative stories more than bare quantitative beliefs.
Textual analysis can elicit the element of causal attribution in political
narratives from these texts (\cite{ash2021text} and \cite{andre2022narratives}
have performed manual and machine anaysis of these texts in order to elicit
prevailing narratives in various contexts). Our formalism, which explicitly
focuses on narratives as simple causal models, will hopefully provide tools
for interpreting such textual data and linking them to economic and social indicators.

\section{Conclusion}

This paper has explored the role of false narratives in the mobilization of
public opinion in heterogeneous societies. Our main insight is that false
narratives enable social groups to dissociate the link between the intrinsic
appeal of certain policies and their adverse outcome. They achieve this by
attributing outcomes to spurious causes, exploiting historical correlations
and misrepresenting them as causal. This typically takes the form of
exclusionary tribal narratives, which argue that keeping certain social groups
out of power leads to good outcomes. Such narratives are reminiscent of
\textquotedblleft scapegoating,\textquotedblright\ a type of narrative that is
often used in the political arena.

This insight suggests a novel perspective into the idea of
\textit{retrospective voting }(see \cite{HM2013} for a review article, and
\cite{PK2017} for an example that extends the concept to multi-party systems).
This is the notion that voters punish or reward parties according to their
performance (measured by certain outcomes) when they were in office. This view
puts less emphasis on processes (i.e., the policies that ruling parties take)
and more emphasis on outcomes. The conventional view is that retrospective
voting is a \textquotedblleft healthy\textquotedblright\ feature of democratic
politics because it improves government accountability and helps selecting
competent candidates. By comparison, our perspective is that attributing
public outcomes to who is (or is not) in power rather than to the implemented
policies can be a demagogic false narrative that is detrimental for public
outcomes. This offers a new and more critical view of retrospective
voting.\pagebreak

\noindent{\LARGE Appendix: Proofs}\bigskip

\noindent\textbf{Proof of Proposition \ref{prop: gen_charact_one_issue}%
\medskip}

\noindent We organize the proof in a series of steps. We will posit the
existence of an essential equilibrium, characterize its properties, and then
confirm that we indeed have an equilibrium. Hereafter, let $\sigma$ be any
candidate essential equilibrium. Note that by definition, $F(a,N)=F(a,N^{a})$.
We use the two notations interchangeably for expositional purposes. For
convenience, hereafter we denote
\begin{equation}
d=F(\ell,N)-F(h,N) \label{d}%
\end{equation}
\medskip

\begin{step}
There exists $(a,C,S)\in Supp(\sigma)$ such that $a=h$.
\end{step}

\begin{proof}
Assume the contrary---i.e., $a=\ell$ for every $(a,C,S)\in Supp(\sigma)$.
Then, $p_{\sigma}(y=1)=0$. Therefore,
\[
U_{\sigma}(a,C,S)=p_{\sigma}(y=1\mid x_{S}(a,C))=0
\]
for every $(a,C,S)\in Supp(\sigma)$. By the definition of equilibrium,
$\sigma$ is the limit of a sequence of $\varepsilon$-equilibria for some
$\varepsilon\rightarrow0$. Since $\sigma(a,C,S)>0$, $\sigma_{\varepsilon
}(a,C,S)$ is bounded away from zero, and therefore $U_{\sigma_{\varepsilon}%
}(a,C,S)\approx p_{\sigma_{\varepsilon}}(y=1\mid x_{S}(a,C))\approx0$, for
some point along the sequence $\varepsilon\rightarrow0$. By contrast,
$U_{\sigma_{\varepsilon}}(h,N^{h},\{0\})=q\cdot F(h,N)$, which is bounded away
from zero and therefore higher than $U_{\sigma_{\varepsilon}}(a,C,S)$. Since
by assumption $(h,N^{h},\{0\})\notin Supp(\sigma)$, we have a
contradiction.\medskip
\end{proof}

This step is the only place in the proof where we use the trembles of
$\varepsilon$-equilibria, to ensure that the payoff from $(h,N^{h},\{0\})$ is
well-defined. From now on, we focus on the $\varepsilon\rightarrow0$ limit
itself.\medskip

\begin{step}
The only $(a,C,S)\in Supp(\sigma)$ with $a=h$ is $(h,N^{h},\{0\})$.
\end{step}

\begin{proof}
For every $(h,C,S)\in Supp(\sigma)$,%
\[
p_{\sigma}(y=1\mid x_{S}(h,C))=q\cdot p_{\sigma}(x_{0}=h\mid x_{S}(h,C))\leq
q
\]
whereas by definition, $p_{\sigma}(y=1\mid x_{0}=h)=q$. Moreover, $F(h,C)\leq
F(h,N^{h})$ by definition. It follows that $U_{\sigma}(h,C,S)\leq U_{\sigma
}(h,N^{h},\{0\})$. Condition $(i)$ in the definition of essential equilibrium
then implies that $S=\{0\}$. Since $F(h,N^{h})>F(h,C^{\prime})$ for every
$C^{\prime}\subset N^{h}$, it follows that $C=N^{h}$.\medskip
\end{proof}

\begin{corollary}
Equilibrium payoffs are
\begin{equation}
U^{\ast}=q\cdot F(h,N^{h}) \label{U*}%
\end{equation}

\end{corollary}

This follows immediately from Steps 1 and 2. Note that the expression for
$U^{\ast}$ is independent of $\sigma$.\medskip From now on, we use the
notation%
\[
\alpha=\sigma(h,N^{h},\{0\})
\]

\begin{step}
If $x_{S}(\ell,C)=x_{S}(h,N^{h})$, then%
\begin{equation}
p_{\sigma}(y=1\mid x_{S}(\ell,C))=\frac{q\alpha}{\alpha+\sum_{C^{\prime
},S^{\prime}\mid x_{S}(\ell,C^{\prime})=x_{S}(\ell,C)}\sigma(\ell,C^{\prime
},S^{\prime})} \label{eq: conditional prob}%
\end{equation}
Otherwise, $p_{\sigma}(y=1\mid x_{S}(\ell,C))=0$.
\end{step}

\begin{proof}
Suppose $0\notin S$. By definition,
\[
p_{\sigma}(y=1\mid x_{S}(\ell,C))=\frac{q\cdot\sum_{C^{\prime},S^{\prime}\mid
x_{S}(h,C^{\prime})=x_{S}(\ell,C)}\sigma(h,C^{\prime},S^{\prime})}%
{\sum_{a^{\prime},C^{\prime},S^{\prime}\mid x_{S}(a^{\prime},C^{\prime}%
)=x_{S}(\ell,C)}\sigma(a^{\prime},C^{\prime},S^{\prime})}%
\]
By Step 2, the numerator can be rewritten as%
\[
q\cdot\sigma(h,N^{h},\{0\})\cdot\mathbf{1}[x_{S}(\ell,C)=x_{S}(h,N^{h})]
\]
which delivers (\ref{eq: conditional prob}). (Note that when $0\notin S$,
$x_{S}(\ell,C)=x_{S}(h,C^{\prime})$ if and only if $S\cap C=S\cap C^{\prime}%
$.) Now suppose $0\in S$. Then,%
\begin{equation}
p_{\sigma}(y=1\mid x_{S}(\ell,C))=p_{\sigma}(y=1\mid x_{0}=\ell)=0
\label{conditional prob 0 in S}%
\end{equation}
\medskip
\end{proof}

\begin{corollary}
For every $(\ell,C,S)\in Supp(\sigma)$, $0\notin S$.
\end{corollary}

\begin{proof}
Suppose $0\in S$. By (\ref{conditional prob 0 in S}), $U_{\sigma}%
(\ell,C,S)=0<U^{\ast}$, hence $(\ell,C,S)\notin Supp(\sigma)$.\medskip
\end{proof}

\begin{step}
If $F(\ell,N)\leq F(h,N)$, then $\alpha=1$. If $F(\ell,N)>F(h,N)$, then%
\[
\alpha\leq\frac{F(h,N)}{F(\ell,N)}%
\]

\end{step}

\begin{proof}
Suppose $F(\ell,N)\leq F(h,N)$, but $\alpha<1$. Then, there exists
$(\ell,C,S)\in Supp(\sigma)$, which implies by definition that the denominator
of (\ref{eq: conditional prob}) is greater than $\alpha$ and hence $p_{\sigma
}(y=1\mid x_{S}(\ell,C))<q$. It follows that
\[
U_{\sigma}(\ell,C,S)=p_{\sigma}(y=1\mid x_{S}(\ell,C))\cdot F(\ell,C)<q\cdot
F(\ell,N)\leq q\cdot F(h,N)=U^{\ast}%
\]
which is a contradiction. Thus, in this case, $\alpha=1$. Suppose
$F(\ell,N)>F(h,N)$. If $\alpha=1$, then
\[
U_{\sigma}(\ell,N^{\ell},\varnothing)=p_{\sigma}(y=1)F(\ell,N)=qF(\ell
,N)>U^{\ast}%
\]
which is again a contradiction. Thus, in this case, $\alpha<1$. Since we must
have $U_{\sigma}(\ell,N^{\ell},\varnothing)\leq U^{\ast}$ in any equilibrium,
and since $p_{\sigma}(y=1)=q\alpha$, it follows that $q\alpha\cdot
F(\ell,N)\leq q\cdot F(h,N)$. This implies that upper bound on $\alpha$ in
this case. Note that this upper bound relies on the assumption that the
denialist narrative is feasible (i.e., $\varnothing\in\mathcal{S}$).\medskip
\end{proof}

The next step establishes that in equilibrium, false narratives take the
\textquotedblleft exclusionary\textquotedblright\ tribal form (including the
denialist narrative as a special case).\medskip

\begin{step}
If $(\ell,C,S)\in Supp(\sigma)$, then $S\subseteq N\setminus N^{h}$ and
$C=N^{\ell}\setminus S$.
\end{step}

\begin{proof}
The proof consists of two steps. First, we show that $N^{h}\cap N^{\ell
}\subseteq C$ for every $(a,C,S)\in Supp(\sigma)$. Then, we show that this
implies the claim. By Step 3, $0\notin S$ for every $(\ell,C,S)\in
Supp(\sigma)$. By the restriction we imposed on the domain of feasible
narratives, $S$ is a weak subset of $N\setminus N^{\ell}$, $N\setminus N^{h}$
or $N^{h}\cap N^{\ell}$. Let us consider each of these cases.

First, suppose $S\subseteq N\setminus N^{\ell}$. By our definition of
admissible coalitions, any platform $(\ell,C,S)$ satisfies $C\cap
S=\varnothing$---i.e., $x_{i}(\ell,C)=0$ for every $i\in S$. By Step~3,
$p_{\sigma}(y=1\mid x_{S}(\ell,C))=0$, hence $U_{\sigma}(\ell,C,S)=0$, a contradiction.

Second, suppose $S\subseteq N\setminus N^{h}$. If $C\subset N^{h}\cap N^{\ell
}$, then $U_{\sigma}(\ell,C,S)<U_{\sigma}(\ell,N^{h}\cap N^{\ell},S)$. The
reason is that $p_{\sigma}(y=1\mid x_{S}(\ell,C))=p_{\sigma}(y=1\mid
x_{S}(\ell,N^{h}\cap N^{\ell}))$ (since the narrative $S$ ignores the power
status of groups in $N^{h}\cap N^{\ell}$), and $f_{i}(\ell)>0$ for every $i\in
N^{h}\cap N^{\ell}$. Since $(\ell,C,S)$ fails to maximize $U_{\sigma}$, we
obtain a contradiction. It follows that $N^{h}\cap N^{\ell}\subseteq C$.

The only remaining case is thus $S\subseteq N^{h}\cap N^{\ell}$. By steps 1
and 2, $x_{i}(h,N^{h})=1$ for every $i\in S$. By Step 3, $p_{\sigma}(y=1\mid
x_{S}(\ell,C))>0$ only if $x_{i}(\ell,C)=1$ for every $i\in S$---i.e.,
$S\subseteq C$. Since $f_{i}(\ell)>0$ for every $i\in N^{h}\cap N^{\ell}$, $C$
is an admissible coalition given policy $\ell$. If $C\subset N^{h}\cap
N^{\ell}$, then $U_{\sigma}(\ell,C,S)<U_{\sigma}(\ell,N^{h}\cap N^{\ell},S)$,
by the same argument as in the previous case. Therefore, $N^{h}\cap N^{\ell
}\subseteq C$.

We have thus established that $S\subseteq N^{\ell}$ for every $(\ell,C,S)\in
Supp(\sigma)$, and that $N^{h}\cap N^{\ell}\subseteq C$ for every platform in
$Supp(\sigma)$. Suppose $S\subseteq N^{h}\cap N^{\ell}$ and $S\neq\varnothing$
for some $(\ell,C,S)\in Supp(\sigma)$. Then,
\[
p_{\sigma}(y=1\mid x_{S}(\ell,C))=p_{\sigma}(y=1\mid x_{i}(\ell,C)=1\text{ for
every }i\in S)=p_{\sigma}(y=1)
\]
Therefore,%
\[
U_{\sigma}(\ell,C,\varnothing)=U_{\sigma}(\ell,C,S)
\]
By condition $(ii)$ in the definition of essential equilibrium, $(\ell
,C,S)\notin Supp(\sigma)$, a contradiction.

The only remaining possibility is that $S\subseteq N\setminus N^{h}$. As we
saw, this implies $N^{h}\cap N^{\ell}\subseteq C$. However, admissibility of
$C$ given $\ell$ requires $C\subseteq N^{\ell}$. It follows that $C=N^{\ell
}\setminus S$.\medskip
\end{proof}

The last step implies that platforms $(\ell,C,S)\in Supp(\sigma)$ are entirely
pinned down by $S$. Therefore, in the rest of the proof, we use the notation
$\bar{\sigma}(S)=\sigma(\ell,C,S)$.\medskip

\begin{step}
\label{lem: eq_ineq} $(\alpha,\bar{\sigma})$ is an equilibrium if and only if,
for all $S\in\mathcal{S}$ that satisfy $S\subseteq N\setminus N^{h}$,
\begin{equation}
\alpha\cdot\frac{d-F(\ell,S)}{F(h,N)}\leq\sum_{S^{\prime}\in\mathcal{S}\mid
S^{\prime}\supseteq S}\bar{\sigma}(S^{\prime}) \label{ineq step 6}%
\end{equation}
\label{eq: eq_ineq} with equality if $\bar{\sigma}(S)>0$. (Recall that $d$ is
defined by (\ref{d}).)
\end{step}

\begin{proof}
Using Definition \ref{def: equilibrium}, $\sigma$ is an equilibrium if and
only if $U_{\sigma}(\ell,C,S)\leq U^{\ast}$ for all $(\ell,C,S)$, with
equality if $\sigma(\ell,C,S)>0$. By Step 5, we can restrict attention to
platforms $(\ell,C,S)$ for which $S\subseteq N\setminus N^{h}$ and $C=N^{\ell
}\setminus S$. Plugging expressions (\ref{U*}) and (\ref{eq: conditional prob}%
), the inequality $U_{\sigma}(\ell,C,S)\leq U^{\ast}$ can be rewritten as a
linear inequality in $\sigma$:%
\[
\alpha\cdot\frac{F(\ell,C)-F(h,N)}{F(h,N)}\leq\sum_{C^{\prime},S^{\prime}\mid
x_{S}(\ell,C^{\prime})=x_{S}(\ell,C)}\sigma(\ell,C^{\prime},S^{\prime})
\]
Note that $F(\ell,C)-F(h,N)=d-F(\ell,S)$, such that the L.H.S of the last
inequality becomes the L.H.S of (\ref{ineq step 6}). Finally, by Step 5,
$(\ell,C^{\prime},S^{\prime})\in Supp(\sigma)$ satisfies $x_{S}(\ell
,C^{\prime})=x_{S}(\ell,C)$ if and only if $S^{\prime}\supseteq S$ (since
$C=N^{\ell}\setminus S$ and $C^{\prime}=N^{\ell}\setminus S^{\prime}$). This
means that we can replace the R.H.S of the last inequality with the R.H.S of
(\ref{ineq step 6}).\medskip
\end{proof}

The last step immediately implies the following observation, which means that
exclusionary tribal narratives will \textquotedblleft
scapegoat\textquotedblright\ a collection of groups only when its aggregate
mobilization potential is not too large.\medskip

\begin{corollary}
\label{step: f_refinement} For every $S\subseteq N\setminus N^{h}$, if
$F(\ell,S)\geq d$, then $\bar{\sigma}(S)=0$.\medskip
\end{corollary}

Inequalities (\ref{ineq step 6}) enable us to construct an algorithm that pins
down the essential-equilibrium distribution.\medskip

\noindent\textbf{An algorithm for computing }$\bar{\sigma}(S)$\textbf{ for all
}$S\in\mathcal{S}$\textbf{ such that }$S\subseteq N\setminus N^{h}$ \noindent
Let%
\[
\overline{\mathcal{S}}=\{S\in\mathcal{S}\mid S\subseteq N\setminus N^{h}\text{
and }F(\ell,S)<d\}.
\]
Define%
\[
\overline{\mathcal{S}}_{1}=\{S\in\overline{\mathcal{S}}\mid\text{there is no
}S^{\prime}\in\overline{\mathcal{S}}\text{ such that }S\subset S^{\prime}\}
\]
Now, for every $k>1$, define $\overline{\mathcal{S}}_{k}$ recursively as
follows:%
\[
\overline{\mathcal{S}}_{k}=\{S\in\overline{\mathcal{S}}\mid\text{there is no
}S^{\prime}\in\overline{\mathcal{S}}\setminus\cup_{j<k}\overline{\mathcal{S}%
}_{j}\text{ such that }S\subset S^{\prime}\}
\]
Since $\overline{\mathcal{S}}$ is finite, in this way we obtain a finite
sequence $\{\overline{\mathcal{S}}_{k}\}_{k=1}^{K}$. This sequence exhausts
all the candidate narratives, as it identifies all exclusionary tribal
narratives allowed by the primitive $\mathcal{S}$. The algorithm starts from
the \textquotedblleft top layer\textquotedblright\ of $\overline{\mathcal{S}}$
(i.e., $\overline{\mathcal{S}}_{1}$) and then proceeds to the other layers in
order. For every $S\in\overline{\mathcal{S}}_{1}$, (\ref{ineq step 6}) can be
written as%
\[
\bar{\sigma}(S)\geq\alpha\cdot\frac{d-F(\ell,S)}{F(h,N)}%
\]
By the definition of $\overline{\mathcal{S}}$, the R.H.S is strictly positive
for every $S\in\overline{\mathcal{S}}_{1}$, which implies that $S$ is in the
equilibrium support and therefore the inequality must hold with equality. This
pins down $\bar{\sigma}(S)$. For every $S\in\overline{\mathcal{S}}$, denote
$\mathcal{H}(S)=\{S^{\prime}\in\overline{\mathcal{S}}\mid S\subset S^{\prime
}\}$. By definition, if $S\in\overline{\mathcal{S}}_{k}$, then $\mathcal{H}%
(S)\subseteq\cup_{j<k}\overline{\mathcal{S}}_{j}$. Now, by induction, suppose
that for all $j<k$ and every $S\in\overline{\mathcal{S}}_{j}$, there exists
$w(S)\geq0$ such that $\bar{\sigma}(S)=\alpha w(S)$. For $S\in\overline
{\mathcal{S}}_{1}$, we have already established that, where $w(S)=(d-F(\ell
,S))/F(h,N)$. For every $S\in\overline{\mathcal{S}}_{k}$, (\ref{ineq step 6})
becomes
\[
\bar{\sigma}(S)=\max\left\{  0\text{ },\text{ }\alpha\cdot\frac{d-F(\ell
,S)}{F(h,N)}-\alpha\sum_{S^{\prime}\in\mathcal{H}(S)}w(S^{\prime})\right\}
\]
where $w(S^{\prime})$ is well-defined for all $S^{\prime}\in\mathcal{H}(S)$,
by the inductive argument. This confirms that $\bar{\sigma}(S)=\alpha w(S)$,
where%
\[
w(S)=\max\left\{  0\text{ },\text{ }\frac{d-F(\ell,S)}{F(h,N)}-\sum
_{S^{\prime}\in\mathcal{H}(S)}w(S^{\prime})\right\}
\]
completing the inductive argument, and thus the definition of the algorithm
for computing $\bar{\sigma}(S)$.$\medskip$

\begin{step}
The algorithm establishes existence and uniqueness of an essential equilibrium.
\end{step}

\begin{proof}
Since $(\alpha,\bar{\sigma})$ must define a probability distribution, we must
have%
\[
\alpha+\sum_{S\in\overline{\mathcal{S}}}\bar{\sigma}(S)=1
\]
Moreover, we have obtained unique expressions for each $\bar{\sigma}(S)$ that
depend multiplicatively on $\alpha$. This pins down the value of $\alpha$:%
\[
\alpha=\frac{1}{1+\sum_{S\in\overline{\mathcal{S}}}w(S)}%
\]
Thus, we have pinned down $(\alpha,\bar{\sigma})$. Since this pair satisfies
all the inequalities (\ref{ineq step 6}), it is therefore an equilibrium by
construction.\bigskip
\end{proof}

\noindent\textbf{Proof of Proposition \ref{prop: eq_categories}\medskip}

\noindent We use the algorithm in the proof of Proposition
\ref{prop: gen_charact_one_issue} to characterize the equilibrium. We begin by
calculating the equilibrium probabilities of the cells that comprise $\Pi$.
Note that $N\setminus N^{h}\in\pi_{1}$ and $F(\ell,N^{h})>F(h,N^{h})$ imply
that $F(\ell,S)<d$ for every $S\subseteq N\setminus N^{h}$ that belongs to one
of the partitions in $\Pi$. Given this, for $N\setminus N^{h}\in\pi_{1}$, we
must have
\[
\bar{\sigma}(N\setminus N^{h})=\alpha\cdot\frac{d-F(\ell,N\setminus N^{h}%
)}{F(h,N)}>0
\]
Let $k>1$. Given $S_{k}\in\pi_{k}$ such that $S_{k}\subset N\setminus N^{h}$,
the collection of sets $\mathcal{H}(S_{k})=\{S^{\prime}\in\overline
{\mathcal{S}}\mid S_{k}\subset S^{\prime}\}$ in the proof of Proposition
\ref{prop: gen_charact_one_issue} takes the form of a chain $\{S_{j}%
\}_{j=1}^{k-1}$ that satisfies $S_{j}\in\pi_{j}$ and $S_{j+1}\subset
S_{j}\subset N\setminus N^{h}$ for all $j<k$. For $S_{2}\in\pi_{2}$, we must
have
\begin{align*}
\bar{\sigma}(S_{2})  &  =\frac{1}{F(h,N)}\max\left\{  0\text{ , }%
\alpha(d-F(\ell,S_{2}))-\alpha(d-F(\ell,S_{1})\right\} \\
&  =\frac{1}{F(h,N)}\max\{0\text{ , }\alpha F(\ell,S_{1}\setminus S_{2})\}>0
\end{align*}
Thus, the coefficient $w(S_{2})$ in the proof of Proposition
\ref{prop: gen_charact_one_issue} takes the form $F(\ell,S_{1}\setminus
S_{2})/F(h,N)$. By induction,
\begin{equation}
\bar{\sigma}(S_{k})=\alpha\frac{F(\ell,S_{k-1}\setminus S_{k})}{F(h,N)}>0
\label{sigma bar Sk}%
\end{equation}
for every $S_{k}\in\pi_{k}$ such that $S_{k}\subset N\setminus N^{h}$ and
$k=2,\ldots,K$. This completes the characterization of the $\bar{\sigma}(S)$
for every cell $S$ in one of the partitions in $\Pi$. Now, consider
$S=\varnothing$. In this case, we need
\begin{equation}
\alpha\frac{d}{F(h,N)}-\sum_{S\in\Pi\mid S\subseteq N\setminus N^{h}}%
\bar{\sigma}(S)\leq\bar{\sigma}(\varnothing) \label{ineq sigma empty}%
\end{equation}
with equality if $\bar{\sigma}(\varnothing)>0$. Suppose that $K=1$ (such that
$R=0$), and hence only $N\setminus N^{h}$ itself belongs to $\Pi$. Then,
\[
\alpha\frac{d}{F(h,N)}-\bar{\sigma}(N\setminus N^{h})=\alpha\frac
{F(\ell,N\setminus N^{h})}{F(h,N)}>0
\]
which implies that
\[
\bar{\sigma}(\varnothing)=\alpha\frac{F(\ell,N\setminus N^{h})}{F(h,N)}>0
\]
If instead $R\geq1$, given our assumptions on $\Pi$ it means that at least two
subsets $S\subset N\setminus N^{h}$ belong to $\pi_{2}$. In this case,%
\begin{align*}
\sum_{S\in\Pi\mid S\subseteq N\setminus N^{h}}\bar{\sigma}(S)  &  \geq
\frac{\alpha}{F(h,N)}\{d-F(\ell,N\setminus N^{h})\}+\sum_{S\in\pi_{2}\mid
S\subseteq N\setminus N^{h}}\bar{\sigma}(S)\\
&  =\frac{\alpha}{F(h,N)}(d-F(\ell,N\setminus N^{h}))\\
&  +\frac{\alpha}{F(h,N)}\sum_{S\in\pi_{2}\mid S\subseteq N\setminus N^{h}%
}F(\ell,N\setminus\{N^{h}\cup S\})\\
&  \geq\frac{\alpha}{F(h,N)}(d-F(\ell,N\setminus N^{h}))+\frac{\alpha}%
{F(h,N)}F(\ell,N\setminus N^{h})\\
&  =\alpha\frac{d}{F(h,N)}%
\end{align*}
The last inequality follows from the fact that since $\pi_{2}$ is a collection
of disjoint sets whose union is $N\setminus N^{h}$, the summation over
$S\in\pi_{2}\mid S\subseteq N\setminus N^{h}$ counts $F(\ell,S)$ at least
once, and potentially more than once. Plugging this inequality in
(\ref{ineq sigma empty}), we obtain $\bar{\sigma}(\varnothing)=0$. It remains
to calculate $\alpha$. Let's start from the case of $R=0$. Plugging the
expressions we obtained for $\bar{\sigma}(N\setminus N^{h})$ and $\bar{\sigma
}(\varnothing)$, we have
\[
1=\alpha+\bar{\sigma}(N\setminus N^{h})+\bar{\sigma}(\varnothing
)=\alpha\left\{  1+\frac{d-F(\ell,N\setminus N^{h})}{F(h,N)}+\frac
{F(\ell,N\setminus N^{h})}{F(h,N)}\right\}
\]
such that%
\[
\alpha=\frac{F(h,N)}{F(\ell,N)}%
\]
For all other cases (where $R>0$), we saw that $\bar{\sigma}(\varnothing)=0$.
For every $S_{k}\in\pi_{k}$ such that $S_{k}\subset N\setminus N^{h}$, let
$S_{k-1}$ be again the antecedent of $S_{k}$ in the chain $\{S_{j}%
\}_{j=1}^{k-1}$ that we used in the construction above. For every $S\in\pi
_{k}$, let $P(S)$ be the unique cell $S^{\prime}\in\pi_{k-1}$ such that
$S\subseteq S^{\prime}$. Given this, and plugging (\ref{sigma bar Sk}), we
have%
\begin{align*}
1  &  =\alpha+\sum_{S\in\Pi\mid S\subseteq N\setminus N^{h}}\bar{\sigma}(S)\\
&  =\frac{\alpha}{F(h,N)}\left\{  F(h,N)+d-F(\ell,N\setminus N^{h})+\sum
_{k=2}^{K}\sum_{S\in\pi_{k}\mid S\subset N\setminus N^{h}}F(\ell,P(S)\setminus
S)\right\} \\
&  =\frac{\alpha}{F(h,N)}\left\{  F(\ell,N^{h})+\sum_{k=2}^{K}\sum_{S\in
\pi_{k}\mid S\subset N\setminus N^{h}}F(\ell,P(S)\setminus S)\right\}
\end{align*}
To further simplify this expression, we now use the assumption that each cell
in $\pi_{k-1}$ has exactly $r_{k}$ subsets in $\pi_{k}$. Using this, we can
rewrite the last condition as
\begin{align*}
1  &  =\frac{\alpha}{F(h,N)}\left\{  F(\ell,N^{h})+F(\ell,N\setminus
N^{h})\sum_{k=2}^{K}(r_{k}-1))\right\} \\
&  =\frac{\alpha}{F(h,N)}\left\{  F(\ell,N^{h})+F(\ell,N\setminus N^{h})\cdot
R\right\}
\end{align*}
such that%
\[
\alpha=\frac{F(h,N)}{F(\ell,N^{h})+F(\ell,N\setminus N^{h})\cdot R}%
\]
This completes the proof.\bigskip

\noindent\textbf{Proof of Proposition \ref{prop: eq_rich_domain}\medskip}

\noindent Let $\sigma$ be the unique essential equilibrium. Since
$F(\ell,N)\geq F(\ell,N^{h})>F(h,N)$, Proposition
\ref{prop: gen_charact_one_issue} implies that $\sigma(h,N^{h},\{0\})=\alpha
\in(0,1)$. Let us now activate the algorithm described in the proof of
Proposition \ref{prop: gen_charact_one_issue}. Recall that we can identify all
the platforms in which $a=\ell$ with the narrative $S\subseteq N\setminus
N^{h}$. Therefore, from now on we will use the abbreviated notation
$\bar{\sigma}(S)=\sigma(\ell,C,S)$ from the proof of Proposition
\ref{prop: gen_charact_one_issue}. Since $F(\ell,N^{h})>F(h,N)$, we have
\[
\bar{\sigma}(N\setminus N^{h})=\alpha\cdot\frac{F(\ell,N^{h})-F(h,N)}%
{F(h,N)}>0
\]
Now consider a narrative $S=N\setminus(N^{h}\cup\{i\})$ for some $i\in
N\setminus N^{h}$. The rich domain assumption means that $S\in\mathcal{S}$. By
definition, $S\nsubseteq S^{\prime}$ for any $S^{\prime}\neq S$ such that
$S^{\prime}\subset N\setminus N^{h}$. Therefore, if $\bar{\sigma}(S)=0$, then
the following inequality must hold:%
\[
\alpha\cdot\frac{d-F(\ell,S)}{F(h,N)}=\alpha\cdot\frac{F(\ell,N^{h}%
\cup\{i\})-F(h,N)}{F(h,N)}\leq\bar{\sigma}(N\setminus N^{h})
\]
which is a contradiction since%
\[
F(\ell,N^{h}\cup\{i\})=F(\ell,N^{h})+f_{i}(\ell)>F(\ell,N^{h}).
\]
It follows that $\bar{\sigma}(N\setminus(N^{h}\cup\{i\}))>0$ for every $i\in
N\setminus N^{h}$, and given by
\[
\bar{\sigma}(N\setminus(N^{h}\cup\{i\}))=\alpha\cdot\frac{f_{i}(\ell)}{F(h,N)}%
\]
We will now show that the support of $\bar{\sigma}$ contains no other
narratives. Assume the contrary. Denote $\left\vert N\setminus N^{h}%
\right\vert =s^{\ast}$. Let $S\subset N\setminus N^{h}$ such that $\left\vert
S\right\vert =s<s^{\ast}-1$. In particular, select $S$ to be a maximal set in
this class---i.e., there exist no other such $S$ that contains it. Then, if
$\bar{\sigma}(S)>0$, it must satisfy the equation
\[
\alpha\cdot\frac{d-F(\ell,S)}{F(h,N)}-\bar{\sigma}(N\setminus N^{h}%
)-\sum_{i\in N\setminus(N^{h}\cup S)}\bar{\sigma}(N\setminus(N^{h}%
\cup\{i\}))=\bar{\sigma}(S)
\]
Plugging our expressions for $\bar{\sigma}(N\setminus N^{h})$ and $\bar
{\sigma}(N\setminus(N^{h}\cup\{i\}))$, the L.H.S becomes
\[
\alpha\cdot\left(  \frac{d-F(\ell,S)}{F(h,N)}-\frac{F(\ell,N^{h}%
)-F(h,N)}{F(h,N)}-\frac{F(\ell,N\setminus(N^{h}\cup S))}{F(h,N)}\right)  =0
\]
Thus, we obtain a contradiction. It remains to obtain the exact value of
$\alpha$. By construction,%
\[
\alpha+\bar{\sigma}(N\setminus N^{h})+\sum_{i\in N\setminus N^{h}}\bar{\sigma
}(N\setminus(N^{h}\cup\{i\}))=1
\]
Plugging our expressions for $\bar{\sigma}(N\setminus N^{h})$ and $\bar
{\sigma}(N\setminus(N^{h}\cup\{i\}))$, we obtain $\alpha=F(h,N)/F(\ell,N)$.
This completes the proof.\bigskip

\noindent\textbf{Proof of Proposition \ref{prop: dynamics}\medskip}

\noindent In this proof, we denote platforms by $z$ whenever convenient to
simplify notation. For every $t$, let $\bar{z}_{t}=(\bar{a}_{t},\bar{C}%
_{t},\bar{S}_{t})\in\arg\max_{z}U_{\sigma_{t}}(z)$ be the dominant platform at
period $t$ and let $\overline{U}_{\sigma_{t}}=U_{\sigma_{t}}(\bar{z}_{t})$ be
the payoff it generates. Note that if there exists $T$ such that $\bar{z}%
_{t}\ne(a,C,S)$ for all $t\ge T$, then $\sigma_{t}(a,C,S)\to0$ as $t\to\infty
$. Recall that $U^{\ast}=q\cdot F(h,N^{h})>0$. The proof proceeds
stepwise.\medskip

\begin{step}
\label{stp: lower bound} $\overline{U}_{\sigma_{t}}\ge U^{*}$ for every $t$.
\end{step}

\begin{proof}
Since $\sigma_{1}$ has full support, $\sigma_{t}(h,N^{h},\{0\})>0$ for every
finite $t$; therefore, $\overline{U}_{\sigma_{t}}\geq U_{\sigma_{t}}%
(h,N^{h},\{0\})=U^{\ast}$ for every $t$.\medskip
\end{proof}

\begin{step}
\label{stp: h_platforms} If $\bar{z}_{t}=(h,C,S)$, then $C=N^{h}$ and
$U_{\sigma_{t}}(h,C,S)=U^{*}$.
\end{step}

\begin{proof}
For every platform $(h,C,S)$ such that $C\subset N^{h}$, $U_{\sigma_{t}%
}(h,C,S)<U_{\sigma_{t}} (h,N^{h},\{0\})$ because $Pr_{\sigma_{t}}(y=1\mid
x_{S}(h,C))\leq q$ and $F(h,C)<F(h,N^{h})$. This also implies that
$U_{\sigma_{t}}(h,N^{h},S)\le U^{*}$ for all $S$ and hence the last
equality.\medskip
\end{proof}

\begin{step}
\label{stp: h_infinitely_often} For all $t$, there exists $t^{\prime}>t$ such
that $\bar{z}_{t^{\prime}}=(h,N^{h},S)$ for some $S$.
\end{step}

\begin{proof}
Step \ref{stp: lower bound} implies that
\[
\underset{t\rightarrow\infty}{\lim\inf}\,\overline{U}_{\sigma_{t}}\geq
U^{\ast}.
\]
Suppose there exists $t$ such that~$\bar{z}_{t^{\prime}}=(\ell,\bar
C_{t^{\prime}},\bar S_{t^{\prime}})$ for all $t^{\prime}\geq t$. This implies
that $Pr_{\sigma_{t}}(y=1\mid x_{\bar{S}_{t}}(\bar{a}_{t},\bar{C}%
_{t}))\rightarrow0$, which is inconsistent with $\lim\inf_{t\rightarrow\infty
}\overline{U}_{\sigma_{t}}>0$.\medskip
\end{proof}

\begin{step}
\label{stp: lim_inf} $\lim\inf\,\overline{U}_{\sigma_{t}}=U^{\ast}.$
\end{step}

\begin{proof}
We have already established that $\lim\inf_{t\rightarrow\infty}\overline
{U}_{\sigma_{t}}\geq U^{\ast}.$ Note that, if $\overline{U}_{\sigma_{t}%
}>U^{\ast}$, then $\bar{z}_{t}=(\ell,C,S)$ for some $C$ and $S$, because
$U_{\sigma_{t}}(h,C^{\prime},S^{\prime})\leq U^{\ast}$ for all $C^{\prime}$
and $S^{\prime}$. Now suppose $\lim\inf_{t\rightarrow\infty}\overline
{U}_{\sigma_{t}}>U^{\ast}$. Then, there exists $T$ such that for all $t\geq
T$, $\bar{z}_{t}$ involves policy $a=\ell$. This contradicts Step
\ref{stp: h_infinitely_often}.\medskip
\end{proof}

Recall that
\[
Pr_{\sigma_{t}}(y=1\mid x_{S}(a,C))=q\cdot\frac{\sum_{C^{\prime},S^{\prime
}\mid x_{S}(h,C^{\prime})=x_{S}(a,C)}\sigma_{t}(h,C^{\prime},S^{\prime})}%
{\sum_{a^{\prime},C^{\prime},S^{\prime}\mid x_{S}(a^{\prime},C^{\prime}%
)=x_{S}(a,C)}\sigma_{t}(a^{\prime},C^{\prime},S^{\prime})}\medskip
\]

\begin{step}
\label{stp: upward_trend} If $\bar{z}_{t}=(h,N^{h},\hat{S})$ and
$x_{S}(h,N^{h})=x_{S}(\ell,C)$, then
\[
Pr_{\sigma_{t+1}}(y=1\mid x_{S}(\ell,C))>Pr_{\sigma_{t}}(y=1\mid x_{S}%
(\ell,C))
\]

\end{step}

\begin{proof}
Given $\bar{z}_{t}=(h,N^{h},\hat{S})$, for every $(\ell,C,S)$ such that
$x_{S}(h,N^{h})=x_{S}(\ell,C)$,%

\begin{align*}
Pr_{\sigma_{t+1}}(y &  =1\mid x_{S}(\ell,C))=q\frac{\frac{1}{t+1}+\frac
{t}{t+1}\sum_{C^{\prime},S^{\prime}\mid x_{S}(h,C^{\prime})=x_{S}(\ell
,C)}\sigma_{t}(h,C^{\prime},S^{\prime})}{\frac{1}{t+1}+\frac{t}{t+1}%
\sum_{a^{\prime},C^{\prime},S^{\prime}\mid x_{S}(a^{\prime},C^{\prime}%
)=x_{S}(\ell,C)}\sigma_{t}(a^{\prime},C^{\prime},S^{\prime})}\\
&  =q\frac{\frac{1}{t}+\sum_{C^{\prime},S^{\prime}\mid x_{S}(h,C^{\prime
})=x_{S}(\ell,C)}\sigma_{t}(h,C^{\prime},S^{\prime})}{\frac{1}{t}%
+\sum_{a^{\prime},C^{\prime},S^{\prime}\mid x_{S}(a^{\prime},C^{\prime}%
)=x_{S}(\ell,C)}\sigma_{t}(a^{\prime},C^{\prime},S^{\prime})}\\
&  >q\frac{\sum_{C^{\prime},S^{\prime}\mid x_{S}(h,C^{\prime})=x_{S}(\ell
,C)}\sigma_{t}(h,C^{\prime},S^{\prime})}{\sum_{a^{\prime},C^{\prime}%
,S^{\prime}\mid x_{S}(a^{\prime},C^{\prime})=x_{S}(\ell,C)}\sigma
_{t}(a^{\prime},C^{\prime},S^{\prime})}\\
&  =Pr_{\sigma_{t}}(y=1\mid x_{S}(\ell,C))
\end{align*}

\end{proof}

\begin{step}
\label{stp: downward_trend} If $\bar{z}_{t}=(\ell,\hat{C},\hat{S})$, then for
every $(\ell,C,S)$,%
\[
Pr_{\sigma_{t+1}}(y=1\mid x_{S}(\ell,C))\leq Pr_{\sigma_{t}}(y=1\mid
x_{S}(\ell,C))
\]
with strict inequality if and only if $x_{S}(\ell,\hat{C})=x_{S}(\ell,C)$.
\end{step}

\begin{proof}
If $\bar{z}_{t}=(\ell,\hat{C},\hat{S})$ and $x_{S}(\ell,\hat{C})\neq
x_{S}(\ell,C)$, then by definition, $Pr_{\sigma_{t+1}}(y=1\mid x_{S}%
(\ell,C))=Pr_{\sigma_{t}}(y=1\mid x_{S}(\ell,C))$. Now suppose that $\bar
{z}_{t}=(\ell,\hat{C},\hat{S})$ and $x_{S}(\ell,\hat{C})=x_{S}(\ell,C)$. Then,%
\begin{align*}
Pr_{\sigma_{t+1}}(y  &  =1\mid x_{S}(\ell,C))=q\frac{\frac{t}{t+1}%
\sum_{C^{\prime},S^{\prime}\mid x_{S}(h,C^{\prime})=x_{S}(\ell,C)}\sigma
_{t}(h,C^{\prime},S^{\prime})}{\frac{1}{t+1}+\frac{t}{t+1}\sum_{a^{\prime
},C^{\prime},S^{\prime}\mid x_{S}(a^{\prime},C^{\prime})=x_{S}(\ell,C)}%
\sigma_{t}(a^{\prime},C^{\prime},S^{\prime})}\\
&  =q\frac{\sum_{C^{\prime},S^{\prime}\mid x_{S}(h,C^{\prime})=x_{S}(\ell
,C)}\sigma_{t}(h,C^{\prime},S^{\prime})}{\frac{1}{t}+\sum_{a^{\prime
},C^{\prime},S^{\prime}\mid x_{S}(a^{\prime},C^{\prime})=x_{S}(\ell,C)}%
\sigma_{t}(a^{\prime},C^{\prime},S^{\prime})}\\
&  <q\frac{\sum_{C^{\prime},S^{\prime}\mid x_{S}(h,C^{\prime})=x_{S}(\ell
,C)}\sigma_{t}(h,C^{\prime},S^{\prime})}{\sum_{a^{\prime},C^{\prime}%
,S^{\prime}\mid x_{S}(a^{\prime},C^{\prime})=x_{S}(\ell,C)}\sigma
_{t}(a^{\prime},C^{\prime},S^{\prime})}\\
&  =Pr_{\sigma_{t}}(y=1\mid x_{S}(\ell,C))\medskip
\end{align*}

\end{proof}

\begin{step}
If $(\ell,C,S)$ is such that $x_{S}(\ell,C)\ne x_{S}(h,N^{h})$, then
$\sigma_{t}(\ell,C,S)\to0$ as $t\to\infty$.
\end{step}

\begin{proof}
Suppose $\sigma_{t}(\ell,C,S)\not \to 0$. Then, there exists a subsequence
such that $\sigma_{t}(\ell,C,S)\to\hat\sigma>0$, which implies that the
denominator of $Pr_{\sigma_{t}}(y=1|x_{S}(\ell,C))$ converges to a strictly
positive number along the subsequence. However, the numerator of
$Pr_{\sigma_{t}}(y=1|x_{S}(\ell,C))$ converges to zero by
Step~\ref{stp: h_platforms}, because $\sigma_{t}(h,C^{\prime},S^{\prime})\to0$
if $x_{S}(h,C^{\prime})=x_{S}(\ell,C)$ and hence $C^{\prime h}$. Therefore,
$U_{\sigma_{t}}(\ell,C,S)\to0$ along the subsequence, which contradicts
$\sigma_{t}(\ell,C,S)\to\hat\sigma>0$.
\end{proof}

\begin{step}
\label{stp: LB_prob} If $(\ell,C,S)$ is such that $x_{S}(\ell,C)=x_{S}%
(h,N^{h})$, then
\[
\underset{t\rightarrow\infty}{\lim\inf}\sum_{C^{\prime},S^{\prime}\mid
x_{S}(h,C^{\prime})=x_{S}(\ell,C)}\sigma_{t}(h,C^{\prime},S^{\prime
})=\underset{t\rightarrow\infty}{\lim\inf}\sum_{S^{\prime}}\sigma_{t}%
(h,N^{h},S^{\prime})\equiv\underline{\sigma}>0
\]

\end{step}

\begin{proof}
The first equality follows because $\sigma_{t}(h,C^{\prime},S^{\prime})\to0$
if $C^{\prime h}$ by Step~\ref{stp: h_platforms} and because $x_{S}%
(\ell,C)=x_{S}(h,N^{h})$. The last inequality is strict because, if
$\underline{\sigma}=0$, there exists a subsequence such that $\sum_{C^{\prime
},S^{\prime}}\sigma_{t}(h,C^{\prime},S^{\prime})\rightarrow0$ and hence
$\sigma_{t}(\ell,C,S)\rightarrow\hat{\sigma}>0$ for some $(\ell,C,S)$ such
that $x_{S}(\ell,C)=x_{S}(h,N^{h})$. However, in this case there exists $T$
such that for all $t\geq T$ in this subsequence the numerator of
$Pr_{\sigma_{t}}(y=1\mid x_{S}(\ell,C))$ becomes arbitrarily small and hence
$U_{\sigma_{t}}(\ell,C,S)<U^{\ast}$, which is inconsistent with $\hat{\sigma
}>0$. \medskip
\end{proof}

\begin{step}
\label{stp: lim_sup} $\lim\sup_{t\rightarrow\infty}\overline{U}_{\sigma_{t}%
}\leq U^{\ast}.$
\end{step}

\begin{proof}
Suppose $\lim\sup_{t\rightarrow\infty}\,\overline{U}_{\sigma_{t}}=\bar
{U}>U^{\ast}$. Let
\[
\bar{P}=\left\{  (\ell,C,S)\mid\lim\sup_{t\rightarrow\infty}\,U_{\sigma_{t}%
}(\ell,C,S)=\bar{U}\right\}  ,
\]
which must be non-empty because the set of platforms is finite. Note that
$(\ell,C,S)\in\bar{P}$ only if $x_{S}(\ell,C)=x_{S}(h,N^{h})$. By finiteness
of $\bar{P}$, there exists a common subsequence, $T$, and $\varepsilon>0$ such
that for all $t^{\prime}\geq T$ in this subsequence $U_{\sigma_{t^{\prime}}%
}(\ell,C,S)\geq U^{\ast}+\varepsilon$ for all $(\ell,C,S)\in\bar{P}$. By Step
\ref{stp: h_infinitely_often}, there must exist a $t>T$ (not necessarily in
the subsequence) such that $\bar{z}_{t}=(h,N^{h},S)$ and hence $\overline
{U}_{\sigma_{t}}=U^{\ast}$. Therefore, $U_{\sigma_{t}}(\ell,C,S)\leq U^{\ast}$
for all $(\ell,C,S)\in\bar{P}$. By Step \ref{stp: upward_trend}, for all
$(\ell,C,S)\in\bar{P}$,%

\begin{align*}
\frac{U_{\sigma_{t+1}}(\ell,C,S)}{U_{\sigma_{t}}(\ell,C,S)}  &  =\frac{\left(
\frac{\frac{1}{t}+\sum_{C^{\prime},S^{\prime}\mid x_{S}(h,C^{\prime}%
)=x_{S}(\ell,C)}\sigma_{t}(h,C^{\prime},S^{\prime})}{\frac{1}{t}%
+\sum_{a^{\prime},C^{\prime},S^{\prime}\mid x_{S}(a^{\prime},C^{\prime}%
)=x_{S}(\ell,C)}\sigma_{t}(a^{\prime},C^{\prime},S^{\prime})}\right)
}{\left(  \frac{\sum_{C^{\prime},S^{\prime}\mid x_{S}(h,C^{\prime})=x_{S}%
(\ell,C)}\sigma_{t}(h,C^{\prime},S^{\prime})}{\sum_{a^{\prime},C^{\prime
},S^{\prime}\mid x_{S}(a^{\prime},C^{\prime})=x_{S}(\ell,C)}\sigma
_{t}(a^{\prime},C^{\prime},S^{\prime})}\right)  }\\
&  <\frac{\left(  \frac{\frac{1}{t}+\sum_{C^{\prime},S^{\prime}\mid
x_{S}(h,C^{\prime})=x_{S}(\ell,C)}\sigma_{t}(h,C^{\prime},S^{\prime})}%
{\sum_{a^{\prime},C^{\prime},S^{\prime}\mid x_{S}(a^{\prime},C^{\prime}%
)=x_{S}(\ell,C)}\sigma_{t}(a^{\prime},C^{\prime},S^{\prime})}\right)
}{\left(  \frac{\sum_{C^{\prime},S^{\prime}\mid x_{S}(h,C^{\prime})=x_{S}%
(\ell,C)}\sigma_{t}(h,C^{\prime},S^{\prime})}{\sum_{a^{\prime},C^{\prime
},S^{\prime}\mid x_{S}(a^{\prime},C^{\prime})=x_{S}(\ell,C)}\sigma
_{t}(a^{\prime},C^{\prime},S^{\prime})}\right)  }\\
&  =\frac{\frac{1}{t}}{\sum_{C^{\prime},S^{\prime}\mid x_{S}(h,C^{\prime
})=x_{S}(\ell,C)}\sigma_{t}(h,C^{\prime},S^{\prime})}+1
\end{align*}
which converges to 1 as $t\rightarrow\infty$ by Step \ref{stp: LB_prob}.
Therefore, for every $\delta>0$, we can pick $T$ large enough such that, for
all $t\geq T$ such that $\bar{z}_{t}=(h,C,S)$,
\[
\frac{U_{\sigma_{t+1}}(\ell,C,S)}{U_{\sigma_{t}}(\ell,C,S)}\leq1+\delta
\]
for all $(\ell,C,S)\in\bar{P}$. Finally, this means that we can also pick $T$
and $t\geq T$ so that $\bar{z}_{t}=(h,C,S)$ and $U_{\sigma_{t+1}}%
(\ell,C,S)<U^{\ast}+\varepsilon$ for all $(\ell,C,S)\in\bar{P}$. Therefore,
$U_{\sigma_{t+k}}(\ell,C,S)<U^{\ast}+\varepsilon$ for all $(\ell,C,S)\in
\bar{P}$ and all $k\geq1$, because by Step \ref{stp: downward_trend} the
payoff of $(\ell,C,S)$ is weakly decreasing when $U_{\sigma_{t}}%
(\ell,C,S)>U^{\ast}$. We, thus, reach a contradiction.\medskip
\end{proof}

Steps~\ref{stp: lim_inf} and \ref{stp: lim_sup} imply that $\lim
_{t\rightarrow\infty}\,\overline{U}_{\sigma_{t}}=U^{\ast}$. Now, denote by
$\Sigma$ the set of limit points of $\sigma_{t}$.\medskip

\begin{step}
All $\sigma\in\Sigma$ must induce the same joint distribution over $(a,C)$,
and this distribution must coincide with the unique equilibrium distribution.
\end{step}

\begin{proof}
Note that $U_{\sigma}(z)$ is continuous in $\sigma$ for all $z$. The previous
conclusion implies that, for every $\sigma\in\Sigma$ and every $z$,
$U_{\sigma}(z)\leq U^{\ast}$, with equality for $z\in Supp(\sigma)$.
Propositions \ref{prop n=2} and \ref{prop: gen_charact_one_issue} established
that every $\sigma$ that satisfies this property induces the same distribution
over $(a,C)$. This completes the proof.\bigskip
\end{proof}

\bibliographystyle{ecta}
\bibliography{false_narravtives}

\end{document}